\documentclass{LMCS}

\def\dOi{10(2:4)2014}
\lmcsheading%
{\dOi}
{1--29}
{}
{}
{Jan.~21, 2011}
{May\phantom.~23, 2014}
{}


\ACMCCS{[{\bf Theory of computation}]: Models of computation---Computability---Lambda calculus}

\usepackage{latexsym}
\usepackage{amsmath,amssymb,amstext,amsthm}
\usepackage{wasysym}

\usepackage{graphicx}
\usepackage{stmaryrd}

\usepackage{caption2}

\usepackage{tikz}
\usetikzlibrary{arrows,automata,backgrounds,calc,chains,decorations.fractals,decorations.pathreplacing,fadings,fit,folding,mindmap,patterns,plotmarks,positioning,shadows,shapes.geometric,shapes.symbols,through,trees}

\colorlet{dblue}{blue!40!black}
\usepackage{hyperref}

\usepackage{textcomp}

\input{macros.sty}

\begin{document}

\title[Discriminating Lambda-Terms Using Clocked B\"{o}hm Trees]
      {Discriminating Lambda-Terms Using Clocked B\"{o}hm Trees\rsuper*}

\author[J.~Endrullis]{J\"{o}rg Endrullis}
\address{%
  VU University Amsterdam,
  Dept.\ of Computer Science,
  De Boelelaan 1081A, 1081 HV Amsterdam
}
\email{\{j.endrullis,r.d.a.hendriks,j.w.klop,a.polonsky\}@vu.nl}

\author[D.~Hendriks]{Dimitri Hendriks}
\address{\vspace{-18 pt}}

\author[J.~W.~Klop]{Jan Willem Klop}
\address{\vspace{-18 pt}}

\author[A.~Polonsky]{Andrew Polonsky}
\address{\vspace{-18 pt}}



\titlecomment{%
  {\lsuper*}This is a modified and extended version of~\cite{endr:hend:klop:2010}
  which appeared in the proceedings of LICS~2010.
  The research has been partially funded by 
  the Netherlands Organisation for Scientific Research (NWO)
  under grant numbers 612.000.934, 639.021.020, and 612.001.002.%
}

\maketitle

\begin{abstract}
As observed by Intrigila~\cite{intri:1997}, there are hardly techniques available 
in the $\lambda$\nb-calculus to prove that two $\lambda$\nb-terms are not $\beta$\nb-convertible. 
Techniques employing the usual \boehm{} Trees are inadequate when we deal with terms 
having the same \boehm{} Tree (BT). 
This is the case in particular for fixed point combinators, as they all have the same~BT. 
Another interesting equation, whose consideration was suggested by Scott~\cite{scott:1975}, 
is \mbox{$\cB\cY = \cB\cY\cS$}, an equation valid in the classical model~$\mcl{P}\omega$ 
of $\lambda$\nb-calculus, and hence valid with respect to BT\nb-equality~$=_{\sbohm}$, 
but nevertheless the terms are $\beta$\nb-inconvertible. 

To prove such $\beta$\nb-inconvertibilities, we employ `clocked' BT's, 
with annotations that convey information of the tempo in which the data in the BT are produced. 
\boehm{} Trees are thus enriched with an intrinsic clock behaviour, 
leading to a refined discrimination method for $\lambda$\nb-terms. 
The corresponding equality is strictly intermediate between $\conv$ and $=_{\sbohm}$, 
the equality in the model $\mcl{P}\omega$.
An analogous approach pertains to \levy{} and \ber{} Trees. 

Our refined \boehm{} Trees find in particular an application in $\beta$\nb-discriminating 
fixed point combinators (fpc's).
It turns out that Scott's equation $\cB\cY = \cB\cY\cS$ is the key to unlocking a plethora of fpc's, 
generated by a variety of production schemes of which the simplest was found by \boehm, 
stating that new fpc's are obtained by postfixing the term $\cS\cI$, also known as Smullyan's Owl.
We prove that all these newly generated fpc's are indeed new, 
by considering their clocked BT's. 
Even so, not all pairs of new fpc's can be discriminated this way. 
For that purpose 
we increase the discrimination power by 
a precision of the clock notion that we call `atomic clock'.


\end{abstract}

\section{Introduction}

\bohm{} Trees constitute a well-known method to discriminate $\lambda$-terms $M$, $N$:
if $\bt{M}$ and $\bt{N}$ are not identical, then $M$ and $N$ are $\beta$-inconvertible, $M \nconv N$.
But how do we prove $\beta$-inconvertibility of $\lambda$-terms with the same \boehm{} Tree?
This question was raised in Scott~\cite{scott:1975} for the interesting equation 
$\comb{B}\comb{Y} = \comb{B}\comb{Y}\comb{S}$ 
between terms that as Scott noted are presumably $\beta$-inconvertible, yet BT-equal ($=_{\sbohm}$).
Scott used his Induction Rule to prove that $\comb{B}\comb{Y} = \comb{B}\comb{Y}\comb{S}$;
instead we will employ below the infinitary $\lambda$-calculus with the same effect,
but with more convenience for calculations as a direct generalization of finitary $\lambda$-calculus.
Often one can solve such a $\beta$-discrimination problem by finding a suitable invariant 
for all the $\beta$-reducts of $M$, $N$\!.
Below we will do this by way of preparatory example for the fixed point combinators (fpc's) in the \bohm{} sequence.
But a systematic method for this discrimination problem has been lacking, 
and such a method is one of the two contributions of this paper.
%
\begin{figure}[h]
\begin{center}
\begin{tikzpicture}[thick,node distance=13mm,inner sep=0.5mm]
  \node (BT) {$\sbohm$};
  \node (cBT) [left of=BT,yshift=-.6cm]{$\scbohm$};
  \node (aBT) [left of=cBT,yshift=-.6cm]{$\sabohm$};

  \node (LLT) [below of=BT] {$\slevi$};
  \node (cLLT) [left of=LLT,yshift=-.6cm]{$\sclevi$};
  \node (aLLT) [left of=cLLT,yshift=-.6cm]{$\salevi$};
  
  \node (BeT) [below of=LLT] {$\sber$};
  \node (cBeT) [left of=BeT,yshift=-.6cm]{$\scber$};
  \node (aBeT) [left of=cBeT,yshift=-.6cm]{$\saber$};

  \node (beta) [below of=aBeT]{$\conv$};

  \draw (BT) -- (cBT) -- (aBT);
  \draw (LLT) -- (cLLT) -- (aLLT);
  \draw (BeT) -- (cBeT) -- (aBeT);

  \draw (BT) -- (LLT) -- (BeT);
  \draw (cBT) -- (cLLT) -- (cBeT);
  \draw (aBT) -- (aLLT) -- (aBeT);
  
  \draw (aBeT) -- (beta);
\end{tikzpicture}
\caption{Comparison of (atomic) clock semantics and unclocked semantics. Higher means more identifications.}
\label{fig:overivew}
\vspace{-4ex}
\end{center}
\end{figure}

Actually, the need for such a strategic method was forced upon us, by the other contribution,
because Scott's equation $\comb{B}\comb{Y} = \comb{B}\comb{Y}\comb{S}$ turned out to be 
the key unlocking a plethora of new fpc's.
The new generation schemes are of the form: 
if $Y$ is an fpc, then $YP_1\ldots P_n$ is an fpc,
abbreviated as $Y \Rightarrow YP_1\ldots P_n$.
So $\cxthole P_1\ldots P_n$ is an `fpc-generating' vector,
and can be considered as a building block to make new fpc's. 
But are they indeed new? 
A well-known example of a (singleton)-fpc-generating vector is $\cxthole\delta$, 
where $\delta = \comb{S}\comb{I}$,
giving rise when starting from Curry's fpc to the \boehm{} sequence of fpc's.
Here another interesting equation is turning up, namely $Y = Y \delta$, 
for an arbitrary fpc $Y$, considered by Statman and Intrigila~\cite{intri:1997}.  
In fact, it is implied by Scott's equation, for an arbitrary fpc $Y$:
\begin{align*}
  \comb{B}Y = \comb{B}Y\comb{S} 
  \,\,\Longrightarrow\,\, 
  \comb{B}Y\comb{I} = \comb{B}Y\comb{S}\comb{I} 
  \,\,\Longleftrightarrow\,\, 
  Y = Y\delta
\end{align*}
The first equation $\comb{B}Y = \comb{B}Y\comb{S}$ will yield many new fpc's, built in a modular way;
the last equation $Y = Y\delta$ addresses the question whether they are indeed new.
Finding ad hoc invariant proofs for their novelty is too cumbersome. 
But fortunately, it turns out that although the new fpc's all have the same BT, namely $\mylam{f}{f^\omega}$,
they differ in the way this BT is formed, in the `tempo of formation', 
where the ticks of the clock are head reduction steps.
We can thus discern a clock-like behaviour of BT's, 
and we refine BT's 
to `clocked' BT's 
by annotating them with this information.
These then enable us to discriminate the terms in question.

This refined discrimination method 
works best for what we call `simple' terms (or, for terms that reduce to simple terms).
A term is `simple' if its reduction to the \boehm{} Tree does not duplicate redexes. 
The class of simple terms is still fairly extensive; 
it includes all fpc's that are constructed in the modular way that we present, 
thereby solving our novelty problem.
In fact, we gain some more ground: though our discrimination method works best 
for pairs of simple terms, it can also fruitfully be applied
to compare a simple term with a non-simple term,
and with some more effort, we can even compare
and discriminate non-simple terms;
see Section~\ref{sec:plotkin} for an example.

Even so, many pairs of fpc's cannot yet be discriminated,
because they not only have the same BT,
they also have the same clocked BT. Therefore, in a final grading up of the precision, 
we introduce `atomic clocks', where the actual position of a head reduction step is administrated. 
All this pertains not only to the BT-semantics, but also to \levy{} Trees (LLT) (or lazy trees), 
and Berarducci Trees (BeT) (or syntactic trees). 
Many problems stay open, in particular problems generalizing the equation of Statman and Intrigila, 
when arbitrary fpc's are considered. 

\subsection*{Overview}
After defining preliminary notions in Section~\ref{sec:prelims}, 
in Section~\ref{sec:fpcs} we are concerned with constructing new fpc's from old, by some generating schemes. 
In Section~\ref{sec:clocked} we define clocked \boehm{} Trees.
The main results here are Theorems~\ref{thm:general} and~\ref{thm:simple}.
The first states that if no reduct of $M$ has a clock that is at least as fast as the clock of $N$,
then $M$ and $N$ are inconvertible, whereas Theorem~\ref{thm:simple} states that 
if $M$ is a simple term then it suffices that the clock $M$ is not eventually faster than the clock of $N$.
In the paper we are mainly concerned with $\lambda$\nb-terms which have simple reducts.
An exception is Section~\ref{sec:plotkin}, where we answer a question of Plotkin (see Related Work below) 
which involves arbitrary (not necessarily simple) fpc's.
Another elaborate example is given in Section~\ref{sec:periodic}, 
where we compute the clocks of three enumerators for Combinatory Logic.
In Section~\ref{sec:atom}, we give a refinement of the clock method 
by not only recording the number of head reduction steps but also their positions.
As an application we show that every combination of the fixed point generating vectors 
$\leftappiterate{\cxthole (\cS\cS)}{\cS}{n} \cI$ introduced in Section~\ref{sec:fpcs} 
give rise to \emph{new} fpc's.
We briefly mention how the theory can be extended to 
the other well-known semantics of $\lambda$\nb-calculus, namely 
\levy{} and \ber{} Trees, in Section~\ref{sec:levy}.
We conclude in Section~\ref{sec:conclusion} with directions for future research.

\subsection*{Related Work}
The present paper is an extension and elaboration of~\cite{endr:hend:klop:2010}.
In particular, as an example application of the main theorems (Theorem~\ref{thm:general}), 
we now answer the following question of Plotkin~\cite{plot:2007}:
\begin{align*}
  \text{\textit{Is there a fixed point combinator $Y$\! such that}}\  Y(\mylam{z}{fzz}) \conv Y(\mylam{x}{Y(\mylam{y}{fxy})}) \punc? 
\end{align*}

An idea similar to clocked B\"{o}hm Trees is employed in 
the excellent paper~\cite{aehl:joac:2002},
pointed out to us by Tarmo Uustalu at the 2010 LICS conference in Edinburgh.
In~\cite{aehl:joac:2002}, Aehlig and Joachimski study continuous normalization 
of the coinductive \mbox{$\lambda$-calculus}
(see~\cite{{kenn:klop:slee:vrie:1997},{joac:2004}})
extended with a void `wait' constructor~$\mcl{R}$.
This extra constructor is returned whenever the head constructor of a term 
cannot immediately be read off from the argument. 
If the head constructor of an application $rs$ is to be found,
it depends on a recursive call to investigate whether $r$ is a variable, 
an abstraction or an application again; in this case an $\mcl{R}$ is returned.
Thus it is guaranteed that the procedure of building a non-wellfounded term 
over the extended grammar
is productive by making the recursion guarded.
Every $\mcl{R}$ is matched either by a $\beta$-step necessary 
to reach the normal form 
(as represented by the B\"{o}hm Tree) 
or by an application node in the B\"{o}hm Tree.
When $r = \mylam{x}{r'}$ then $\mcl{R}$ matches a $\beta$-step 
and so $s$ is used in the substitution $\subst{r'}{x}{s}$.
Or we find that $r$ is a variable and 
so $\mcl{R}$ matches an application node in the B\"ohm Tree.
So the number of $\mcl{R}$'s in the normal form of a term $t$ is precisely related 
to the number of reduction steps and the `size' of the resulting B\"{o}hm Tree.
Summarizing, similar to the clocked B\"ohm Trees that we introduce in Section~\ref{sec:clocked}, 
in~\cite{aehl:joac:2002} B\"{o}hm Trees are enriched with information about 
how the tree is constructed. 
However, in~\cite{aehl:joac:2002} the refinement is used for proving 
their normalization function to be continuous, 
and not for discriminating \mbox{$\lambda$-terms}, 
the goal we pursue here.

In~\cite{gold:1995} a heuristic procedure in finitary $\lambda$-calculus is
given to construct fpc's in a uniform way.

\section{Preliminaries}\label{sec:prelims}

To make this paper moderately self-contained,
and to fix notations, we lay out some ingredients.
For $\lambda$-calculus we refer to \cite{bare:1984} and \cite{beth:2003}.
For an introduction to \boehm{}, \ber{} and \levy{} Trees, 
we refer to \cite{bare:1984,abra:ong:1993,beth:klop:vrij:2000,bare:klop:2009}.

\begin{defi}\label{def:terms}
  Let $\mcl{X}$ be an countably infinite set of variables. 
  The set $\lterm$ of \emph{finite $\lambda$-terms} is defined inductively 
  by the following grammar:
  \begin{align*}
    M & \BNFis x \BNFor \mylam{x}M \BNFor M \cdot M && (x \in \mcl{X})
  \end{align*}
  We use $x, y, z, \ldots$ for variables, and $M,N,\ldots$ to range over the elements of~$\lterm$.
\end{defi}
Thus $\lambda$-terms are variables, abstractions or applications.
Usually we suppress the application symbol in a term $M\cdot N$
and just write $(MN)$. 
We also adopt the usual conventions for omitting brackets,
i.e., we let application associate to the left, 
so that $N_1 N_2 \ldots N_k$ denotes $(\ldots(N_1 N_2) \ldots N_k)$,
and we let abstraction associate to the right: 
$\mylam{x_1\ldots x_n}{M}$ stands for $(\mylam{x_1}{(\ldots(\mylam{x_n}{(M)}))})$.

\begin{defi}\label{def:infinite:terms}
  The set $\infterm{\lambda}$ of (finite and) \emph{infinite $\lambda$-terms} is defined 
  by interpreting the grammar from Definition~\ref{def:terms}
  \emph{coinductively}, that is, 
  $\infterm{\lambda}$ is the largest set $X$ such that 
  every element $M \in X$ is either a variable $x$, an abstraction $\mylam{x}{M'}$ or an application $M_1 M_2$
  with $M',M_1,M_2\in X$.
  (See further~\cite{sang:rutt:2012} for a precise treatment of coinductive definition and proof principles.)
\end{defi}

\boehm{}, \levy{} and \ber{} Trees are infinite $\lambda$-terms
with an additional clause for $\bot$ in the grammar, 
where $\bot$ stands for the different notions of `undefined' in these semantics,
see Definitions~\ref{def:cbohm}, \ref{def:clevi} and \ref{def:cber}.

\begin{defi}
  The relation $\to_\beta$ on $\term{\lambda}$ or $\infterm{\lambda}$, called \emph{$\beta$-reduction},  
  is the compatible closure (i.e., closure under term formation)
  of the $\beta$-rule:
  \begin{align}
    (\mylam{x}{M})N & \to \subst{M}{x}{N}
    \tag{$\beta$}
  \end{align}
  where $\subst{M}{x}{N}$ denotes the result of substituting $N$ 
  for all free occurrences of $x$ in $M$.
  Furthermore, we write
  $\to_{\beta}^{n}$ for the $n$-fold composition of $\to_{\beta}$, 
  defined by $M\to_{\beta}^{0}M$, 
  and $M\to_{\beta}^{n+1}N$ if $M\to_{\beta}P$ and $P\to_{\beta}^{n}N$ for some $P\in\lterm$.
  We use $\smred_{\beta}$ to denote the reflexive-transitive closure of $\to_{\beta}$,
  $M\mred_{\beta} N$ iff $\myex{n}{M\to_{\beta}^n N}$;
   we let $\to_\beta^{=}$ denote the reflexive closure of $\to_\beta$, 
  ${\to_\beta^{=}} = {\to_\beta^{0}} \cup {\to_\beta^1}$.
  We write $M \conv N$ to denote that $M$ is $\beta$-convertible with $N$,
  i.e., $\sconv$ is the equivalence closure of $\to_\beta$.
  For syntactic equality (modulo renaming of bound variables), we use $\equiv$.
\end{defi}

Apart from the `many-steps' relation $\mred_\beta$,
we introduce `multi-steps' $\dred_\beta$~\cite{{terese:2003},{bare:1984}}, 
which arise from complete developments.
\begin{defi}
  We write $\dred_\beta$ for \emph{multi-steps}, that is, 
  complete developments of a set of redex occurrences in a term.
  A \emph{development} of a set of redex occurrences $U$ in a term contracts exclusively descendants (residuals) of $U$, 
  and it is called \emph{complete} if there are no residuals of $U$ left.
  %
\end{defi}  
A complete development of $U$ can alternatively be viewed as contracting 
all redex occurrences in $U$ in an inside-out fashion.

We will often omit the subscript $\beta$ in $\to_\beta$, $\smred_\beta$ and $\dred_\beta$. 

\begin{defi}
  A $\lambda$-term $M$ is called a \emph{normal form} if there exists no $N$ with $M \to N$.
  We say that a term $M$ \emph{has a normal form} if it reduces to one.
  For $\lambda$-terms $M$ having a normal form we write
  $\nf{M}$ for the unique normal form $N$ with $M \mred N$ 
  (uniqueness follows from confluence of the $\lambda$-calculus).
\end{defi}

Some commonly used combinators are:
\begin{align*}
  \comb{I} & \defeq \mylam{x}{x} &
  \comb{K} & \defeq \mylam{xy}{x} &
  \comb{S} & \defeq \mylam{xyz}{xz(yz)} &
  \comb{B} & \defeq \mylam{xyz}{x(yz)}
\end{align*}

\begin{defi}\label{def:pos}
  A \emph{position} is a sequence over $\{0,1,2\}$.
  For a finite or infinite $\lambda$-term~$M$, 
  the \emph{subterm $\subtrmat{M}{\apos}$ of $M$ at position $\apos$}
  is defined by:
  \begin{align*}
    \subtrmat{M}{\posemp} &= M
    &
    \subtrmat{(MN)}{\posconcat{1}{\apos}} 
    & = \subtrmat{M}{\apos} 
    \\
    \subtrmat{(\mylam{x}{M})}{\posconcat{0}{\apos}} 
    & = \subtrmat{M}{\apos} 
    &
    \subtrmat{(MN)}{\posconcat{2}{\apos}} 
    & = \subtrmat{N}{\apos}
  \end{align*}
  $\pos{M}$ is the set of positions $\apos$ such that $\subtrmat{M}{\apos}$ is defined.
\end{defi}


\begin{defi}
  \hfill
  \begin{enumerate}
    \item 
      A term $\afpc$ 
      is an \emph{fpc} 
      if $\afpc x \conv x(\afpc x)$.
    \item 
      An fpc $\afpc$ is \emph{$k$-reducing} 
      if $\afpc x \redn{k} x(\afpc x)$.
    \item 
      An fpc $\afpc$ is \emph{reducing} 
      if $\afpc$ is $k$-reducing for some $k \in \nat$.
    \item\label{item:wfpc}
      A term $\awfpc$ is a \emph{weak fpc (wfpc)} 
      if $\awfpc x \conv x (\awfpc' x)$
      where $\awfpc'$ is again a wfpc.
  \end{enumerate}
\end{defi}
\noindent
The definition of weak fpc's in~\eqref{item:wfpc} is essentially coinductive~\cite{sang:rutt:2012}, 
that is, implicitly employing a `largest set' semantics.
In long form, the definition means the following:
the set of weak fpc's is the largest set $W \subseteq \lterm$ such that
for every $Z \in W$ we have $\awfpc x \conv x (\awfpc' x)$ for some $Z' \in W$.

A wfpc is alternatively defined as a term having the same B\"{o}hm Tree
as an fpc, namely $\mylam{x}{x^{\omega}} \equiv \mylam{x}{x(x(x(\ldots)))}$.
Weak fpc's are known in foundational studies of type systems 
as \emph{looping combinators}; see, e.g., \cite{coqu:herb:1994} and \cite{geuv:wern:1994}.
\begin{exa}
  Define by double recursion, $Z$ and $Z'$ such that 
  $Zx = x(Z'x)$ and $Z'x = x(Zx)$.
  Then $Z,Z'$ are both wfpc's, and $Zx = x(x(Zx))$. 
  So $Z$ delivers its output twice as fast as an ordinary fpc, but the generator flipflops.
\end{exa}
As to `double recursion', \cite{klop:2007}
collects several proofs of the double fixed point theorem, 
including some in~\cite{bare:1984,smull:1985}.

\begin{defi}
  \hfill
  \begin{enumerate}
  \item 
  A \emph{head reduction step $\hred$} is a $\beta$-reduction step of the form:\\
  $\mylam{x_1\ldots x_n}{(\mylam{y}{M})N N_1 \ldots N_m} 
   \to \mylam{x_1\ldots x_n}{(\subst{M}{y}{N})N_1 \ldots N_m}$
  with $n,m \ge 0$.
  \item 
    Accordingly, a \emph{head normal form (hnf)} is a $\lambda$-term
    of the form:\\
    $\mylam{x_1}{\ldots\mylam{x_n}{y N_1 \ldots N_m}}$ with $n, m \ge 0$.
  \item
    A \emph{weak head normal form (whnf)} is 
    an hnf or an abstraction,
    that is, a whnf is a term of the form $x M_1 \ldots M_m$ or $\mylam{x}{M}$.
  \item
    A term \emph{has a (weak) hnf} if it reduces to one.
  \item
    We call a term \emph{root-stable} 
    if it does not reduce to a redex: $(\mylam{x}{M}) N$.
    A term is called \emph{root-active} if it does not reduce
    to a root-stable term.
  \end{enumerate}
\end{defi}

\subsection*{Infinitary \texorpdfstring{$\lambda$}{lambda}-calculus
  \texorpdfstring{$\mbs{\inflamcal}$}{lambda infinity beta}}
We will only use the infinitary $\lambda$-calculus $\inflamcal$ for some simple calculations 
such as $(\mylam{ab}{(ab)^\omega}) \comb{I} =_{\inflamcal} \mylam{b}{(\comb{I} b)^\omega} =_{\inflamcal} \mylam{b}{b^\omega}$.
Here $M^\omega$ denotes the infinite $\lambda$-term $M(M(M(\ldots)))$
obtained as the solution of $M^\omega = M\,M^\omega$.
For a proper setup of $\inflamcal$ we refer to
\cite{bera:intr:1996,kenn:klop:slee:vrie:1997,kenn:vrie:2003,bare:klop:2009}.
%
In a nutshell, 
$\inflamcal$ extends finitary $\lambda$-calculus by admitting infinite $\lambda$-terms, 
the set of which is called $\infterm{\lambda}$, 
and infinite reductions (in~\cite{kenn:vrie:2003,bare:klop:2009} 
possibly transfinitely long, in~\cite{bera:intr:1996} of length $\leq \omega$).
Limits of infinite reduction sequences are obtained by a strengthening of Cauchy-convergence, 
stipulating that the depth of contracted redexes must tend to infinity.
The $\inflamcal$-calculus is not infinitary confluent ($\CRinf$), 
but still has unique infinite normal forms ($\UNinf$). 
\boehm{} Trees (BT's) without $\bot$ are infinite normal forms in $\inflamcal$.
But beware, the reverse does not hold, e.g.\ $\mylam{x}{(\mylam{x}{(\mylam{x}{\ldots})})}$
is an infinite normal form, but not a BT; it is in fact an LLT (\levy{} Tree, 
and also a BeT (Berarducci Tree). 
The notions BT, LLT, BeT are defined e.g.\ in~\cite{bare:klop:2009},
and in~\cite{beth:klop:vrij:2000}.
These notions are also defined in Sections~\ref{sec:clocked} and~\ref{sec:levy}, 
via their clocked versions.

\begin{defi}
  For terms $A,B$ we define
  $\leftappiterate{A}{B}{n}$ and $\rightappiterate{A}{n}{B}$:
  \begin{align*}
  \leftappiterate{A}{B}{0}    & = A 
  &&& 
  \rightappiterate{A}{0}{B}   & = B 
  \\
  \leftappiterate{A}{B}{n+1}  & = \leftappiterate{\app{A}{B}}{B}{n} 
  &&& 
  \rightappiterate{A}{n+1}{B} & = \app{A}{(\rightappiterate{A}{n}{B})}
  \end{align*}
  A context of the form $\leftappiterate{\cxthole}{B}{n}$ 
  is called a \emph{vector}.
  For the vector notation, it is to be understood that 
  term formation gets highest priority,
  i.e., 
  $\leftappiterate{AB}{C}{n} = \leftappiterate{(AB)}{C}{n}$.
\end{defi}

\section{Fixed Point Combinators}\label{sec:fpcs}

\begin{quotation}
  \noindent
  \textit{%
    The theory of sage birds (technically called \emph{fixed point combinators})
    is a fascinating and basic part of combinatory logic;
    we have only scratched the surface.%
  }
  \\
  \hspace*{\fill} R.~Smullyan~\cite{smull:1985}.
\end{quotation}

\subsection{The \boehm{} Sequence}\label{sec:boehm}

There are several ways to make fpc's. 
For heuristics behind the construction of Curry's fpc 
$\fpcC \defeq \mylam{f}{\omega_{f}\omega_{f}}$,
with $\omega_{f} \defeq \mylam{x}{f(xx)}$, 
and Turing's fpc $\fpcT \defeq \eta\eta$
with  $\eta \defeq \mylam{xf}{f(xxf)}$,
see~\cite{bare:1984,klop:2007}.
%

%
It is well-known, as observed by C.~B\"{o}hm~\cite{bohm:1963,bare:1984}, 
that the class of fpc's coincides exactly with the class of fixed points 
of the peculiar term $\delta = \mylam{ab}{b(ab)}$, convertible with $\comb{S}\comb{I}$. 
The notation $\delta$ is convenient for calculations 
and stems from~\cite{intri:1997}.
This term also attracted the attention of R.~Smullyan,
in his beautiful fable about fpc's figuring as birds 
in an enchanted forest: 
  ``An extremely interesting bird is the owl $O$ defined
    by the following condition: $Oxy = y(xy)$.'' 
  \cite
  {smull:1985}. 
We will return to the Owl in Remark~\ref{rem:owl} below.

Thus the term $Y \delta$ is an fpc whenever $Y$ is.
It follows that starting with $\fpc{0}$, Curry's fpc,
we have an infinite sequence of fpc's
$\fpc{0}, \fpc{0}\delta, \fpc{0}\delta\delta, \ldots, \leftappiterate{\fpc{0}}{\delta}{n}, \ldots$;
we call this sequence the `\boehm{} sequence'.
Note that $\fpc{0}\delta \conv \eta\eta$,
justifying the overloaded notation $\fpc{1}$. 
Now the question is whether all these `derived' fpc's are really new, 
in other words, whether the sequence is free of duplicates. 
This is {*}Exercise~6.8.9 in~\cite{bare:1984}. 

Note that we could also have started the sequence from another fpc than Curry's. 
Now for the sequence starting from an \emph{arbitrary} fpc $Y$, 
it is actually an open problem whether that sequence of fpc's 
$Y, Y\delta, Y\delta\delta, \ldots, \leftappiterate{Y}{\delta}{n}, \dots$ is free of repetitions. 
All we know, applying Intrigila's theorem, Theorem~\ref{thm:intrigila} below, 
is that no two consecutive fpc's 
in this sequence are convertible. But let us first consider the B\"{o}hm sequence. 

\begin{defi}\label{def:boehm:sequence}
  The \emph{\boehm{} sequence} is the sequence $(\fpc{n})_{n \ge 0}$ where $\fpc{n}$ is defined by 
  \begin{align*}
    \fpc{0} & =  \mylam{f}{\omega_f\omega_f}
    & \fpc{n} & = \leftappiterate{\eta\eta}{\delta}{(n-1)} \quad (n > 0)
  \end{align*}
\end{defi}
We show that the \boehm{} sequence contains no duplicates by 
determining the set of reducts of every $\fpc{n}$.
For $\fpc{3}$\footnote{%
  Actually Figure~\ref{fig:Y3:head:reduction:graph} displays $\fpc{3} x$.
  We will frequently consider $\afpc x$ instead of $\afpc$ as then the repetition is immediate,
  and because we have that if $M x \nconv M' x$ then also $M \nconv M'$.%
}
the head reduction is displayed in 
Figure~\ref{fig:Y3:head:reduction:graph},
but this is by no means the whole reduction graph.
For future reference we note that the head reduction diagram suggests a `clock behaviour'.
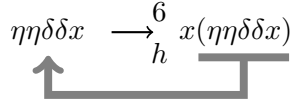
\begin{figure}[ht!]
  \begin{center}
  \begin{tikzpicture}[thick]
  \node (Y3) {$\eta \eta \delta \delta x$};
  \node (R) [right of=Y3,node distance=25mm] {$x(\eta \eta \delta \delta x)$};
  \draw [->,shorten >= 2mm,shorten <= 2mm] (Y3) -- (R) node [very near end,below] {$h$} node [very near end,above] {$6$};
  \draw [gray,line width=3pt] ($(R.south west)+(4mm,.0mm)$) -- ($(R.south east)+(-2mm,.0mm)$);
  \draw [->,gray,line width=3pt] ($(R.south west)!.5!(R.south east)+(1mm,.0mm)$) -- +(0mm,-.5cm) -| ($(Y3.south west)!.5!(Y3.south east)$);
  \end{tikzpicture}
  \caption{Head reduction of $\fpc{3} x$.} 
  \label{fig:Y3:head:reduction:graph}
  \end{center}
\end{figure}

\begin{thm}\label{thm:boehm:seq}
  The B\"{o}hm sequence $(\fpc{n})_{n \ge 0}$ contains no duplicates.
\end{thm}
\begin{proof}
  We define languages $\lang{n}\subseteq\lterm$ as follows:
  \begin{align*}
  \lang{0} & \BNFis 
    \mylam{f}{\rightappiterate{f}{k}{(\omega_f\omega_f)}}
    && (k \geq 0)
  \\
  \lang{1} & \BNFis 
    \eta\eta %
    \BNFor 
    \mylam{f}{\rightappiterate{f}{k}{(\lang{1} f)}} %
    && (k > 0)
  \\
  \lang{n} & \BNFis 
    \lang{n-1}\delta 
    \BNFor 
    \mylam{b}{b^k (\lang{n}b)}
    \BNFor 
    \delta \lang{n}
    && (n > 1, k > 0)
  \end{align*}
  We show that:
  \begin{enumerate}
  \item
    $\fpc{n} \in \lang{n}$;
  \item
    $\lang{n}$ is closed under \mbox{$\beta$-reduction}; and
  \item
      $\lang{n}$ and $\lang{m}$ are disjoint, for $n \neq m$.
  \end{enumerate}
  Then it follows that $\lang{n}$ contains the set of $\smred_\beta$-reducts of $\fpc{n}$,
  and using (iii) this implies that $\fpc{n} \nconv \fpc{m}$
  for all $n\ne m$.
  
  For (i) note that $\fpc{0} \in \lang{0}$, 
  and 
  $\fpc{n} = \leftappiterate{\eta\eta}{\delta}{(n-1)}$
  which is in $\lang{n}$ by induction on $n > 0$. 
  
  We show (ii): if $M\in\lang{n}$ and $M \to N$, then $N\in\lang{n}$.
  Using induction, we do not need to consider cases where the rewrite
  step is inside a variable of the grammar.
  We write $\lang{n}$ in terms as shorthand for
  a term $M \in \lang{n}$.
  \begin{enumerate}[label=($\lang{\arabic*}$),start=0]
    \item 
    We have $\mylam{f}{\rightappiterate{f}{k}{(\omega_f\omega_f)}} 
    \to \mylam{f}{\rightappiterate{f}{k+1}{(\omega_f\omega_f)}} \in \lang{0}$.

    \item 
    We have $\eta\eta \to \mylam{f}{f(\eta\eta f)} \in \lang{1}$,\\
    and $\mylam{f}{\rightappiterate{f}{{k}}{(\mylam{f}{\rightappiterate{f}{{\ell}}{(\lang{1} f)}} f)}}
    \to \mylam{f}{\rightappiterate{f}{{k+\ell}}{(\lang{1} f)}} \in \lang{1}$.

    \item [($\lang{n}$)]
    \noindent
    Case 1:
    $\mylam{f}{\rightappiterate{f}{{k}}{(\lang{1} f)}} \delta 
    \to \rightappiterate{\delta}{{k}}{(\lang{1} \delta)} \in \lang{n}$
    for $n = 2$,
    and
    $(\mylam{b}{b^k(\lang{n-1}b)}) \delta \to \delta^k (\lang{n-1} \delta)$ $\in \lang{n}$ 
    for $n > 2$.
    
    \noindent
    Case 2: 
    $\mylam{b}{b^k(\mylam{c}{c^\ell(\lang{n} c)} b)} \to \mylam{b}{b^{k+\ell}(\lang{n} b)} \in \lang{n}$.

    \noindent
    Case 3: $\delta \lang{n} \to \mylam{b}{b(\lang{n} b)} \in \lang{n}$. 
  \end{enumerate}

  \noindent For $n \neq m$, $n > 1$, 
  (iii) follows by counting the number of passive $\delta$'s,
  that is, the number of occurrences of the form $P \delta$ for some $P$.
  To see that $\lang{0} \cap \lang{1} = \setemp$,
  note that if $M \in \lang{1}$ 
  is an abstraction,
  then $M \equiv \mylam{f}{\rightappiterate{f}{k}{(P f)}}$
  containing a subterm $P f$ which is never the case in~$\lang{0}$.
\end{proof}

A very interesting theorem involving $\delta$ was proved by B.~Intrigila,
affirming a conjecture by R.~Statman. 

\begin{thm}[Intrigila~\cite{intri:1997}]\label{thm:intrigila}
  There is no `double' fixed point combinator.
  That is, for no fpc $\afpc$\! we have $\afpc\delta \conv \afpc$.
\end{thm}

\begin{rem}[Smullyan's Owl $\comb{S}\comb{I} \conv \delta \defeq \mylam{ab}{b(ab)}$]\label{rem:owl}\mbox{}\\
  We collect some salient facts and questions.
  \begin{enumerate}
    \item
      If $Z$ is a wfpc, both $\delta Z$ and $Z \delta$ are wfpc's \cite{smull:1985}.
    \item
      Call an applicative combination of $\delta$'s a \mbox{\emph{$\delta$-term}},
      that is, a term solely built from $\delta$ and application.
      In spite of $\delta$'s simplicity, not all \mbox{$\delta$-terms} 
      are strongly normalizing (SN). 
      An example of a \mbox{$\delta$-term} with infinite reduction 
      is $\delta\delta(\delta\delta)$ 
      (Johannes Waldman, Hans Zantema, personal communication, 2007).
    \item 
      Let $M$ be a non-trivial $\delta$-term, i.e., not a single $\delta$.
      Then $M$ is SN iff $M$ contains exactly one occurrence of $\delta\delta$.
      Furthermore, if $\delta$-terms $M,M'$ are SN, then they are convertible if and only if 
      $M,M'$ have the same length~\cite{klop:2007}. 
      Here the length of a $\delta$-term is the number $\delta$'s in the term.
    \item 
      Convertibility is decidable for $\delta$-terms \cite{stat:1989}.
    \item
      We define $\Delta = \delta^\omega$, so $\Delta \equiv \delta\Delta$.
      Then, the infinite $\lambda$-term $\Delta$ is an fpc:
      $\Delta x \equiv \delta\Delta x \mred x (\Delta x)$.
      The term $\Delta$ can be normalized: $\Delta \to_{\omega} \mylam{f}{f^\omega}$.
      There are many more infinitary fpc's, e.g., for every $n$,
      the infinite term $\leftappiterate{(\comb{S}\comb{S})^\omega}{\comb{S}}{n} \comb{I}$ is one,
      as will be clear from the sequel.
    \item 
      The term $\delta\delta(\delta\delta)$ has no hnf, 
      and hence its \bohm{} Tree is trivial, $\sbohm(\delta\delta(\delta\delta)) \equiv \sink$.
      However, its \ber{} Tree is not trivial.
      Zantema remarked that $\delta$-terms, even infinite ones,
      such as $\Delta\Delta$, are ``top-terminating'' 
      (Zantema considered the applicative rule for $\delta$ only ---
      we expect that his observation remains valid for the \mbox{$\lambda\beta$-version}).
    \item 
      Is Intrigila's theorem also valid for wfpc's:
      for no wfpc $\awfpc$ we have $\awfpc \delta \conv \awfpc$?
  \end{enumerate}
\end{rem}

\subsection{The Scott Sequence}\label{sec:scott}

In \cite[p.~360]{scott:1975} the equation $\comb{B}\fpcC = \comb{B}\fpcC\comb{S}$ 
is mentioned as an interesting example of
an equation not provable in $\lamcal$%
  \footnote{This equation is also discussed in~\cite{deza:seve:vrie:2003}.},
while easily provable with Scott's Induction Rule.
Scott mentions that he expects that using `methods of \boehm' 
the non-convertibility in $\lamcal$ can be established, 
but that he did not attempt a proof. 
On the other hand, with the induction rule the equality is easily established. 
We will not consider Scott's Induction Rule, but we will be working 
in the infinitary lambda calculus, $\inflamcal$.
It is readily verified that in $\inflamcal$ we have:
\[
  \comb{B} \fpcC =_{\inflamcal} \comb{B} \fpcC \comb{S} =_{\inflamcal} \mylam{ab}{(ab)^{\omega}}
\]

\begin{prop}\label{prop:BY:neq:BYS}
  $\comb{B}\fpcC \notconv \comb{B} \fpcC \comb{S}$
\end{prop}

\begin{proof}
  Postfixing the combinator $\comb{I}$ yields $\comb{B}\fpcC \comb{I}$ and $\comb{B}\fpcC \comb{S}\comb{I}$.
  Now $\comb{B} \fpcC \comb{I} \conv \fpcC$ and $\comb{B} \fpcC \comb{S}\comb{I} \conv \fpcC(\comb{S}\comb{I}) = \fpcT$.
  Because $\fpcC \nconv \fpcT$ (Theorem~\ref{thm:boehm:seq}), 
  the result follows.
  \end{proof}
\noindent
In the same way we can strengthen this non-equation to all fpc's $Y$, 
using 
Theorem~\ref{thm:intrigila}. 

\begin{rem}
  \hfill
  \begin{enumerate}
  \item
  The idea of postfixing an $\comb{I}$ is suggested by the \bohm{} Tree $\mylam{ab}{(ab)^{\omega}}$
  of $\comb{B}Y$ and $\comb{B}Y\comb{S}$. Namely, in $\inflamcal$ we calculate:
  $(\mylam{ab}{(ab)^{\omega}}) \comb{I} = \mylam{b}{(\comb{I} b)^{\omega}} = \mylam{b}{b^{\omega}}$
  which is the \bohm{} Tree of any fpc.
  \item
    Interestingly, Scott's equation $\comb{B} Y = \comb{B} Y \comb{S}$
    implies the equation of Statman and Intrigila,
    $Y = Y \delta$
    as one readily verifies, as in the proof of Proposition~\ref{prop:BY:neq:BYS}.
  \end{enumerate}
\end{rem}

\noindent Actually, the comparison between the terms $\comb{B}Y$ and $\comb{B}Y\comb{S}$ has more
in store for us than just providing an example that the extension
from finitary lambda calculus $\lambda\beta$ to infinitary lambda
calculus $\inflamcal$ is not conservative. 
The BT-equality of $\comb{B}Y$ and $\comb{B}Y\comb{S}$ suggests looking 
at the whole sequence
$\comb{B}Y, \comb{B}Y\comb{S}, \comb{B}Y\comb{S}\comb{S}, \ldots, \leftappiterate{\comb{B}Y}{\comb{S}}{n},\ldots$.
All these terms 
have the \bohm{} Tree $\lambda ab.(ab)^{\omega}$,
and hence they are not fpc's.
But 
postfixing an $\comb{I}$ turns them into fpc's. 

\begin{defi}
  The \emph{Scott sequence} is defined by:
  \[ 
    \comb{B}\fpcC \comb{I},\;
    \comb{B}\fpcC \comb{S}\comb{I} ,\;
    \comb{B} \fpcC \comb{S}\comb{S}\comb{I},\;
    \ldots,\;
    \leftappiterate{\comb{B}\fpcC}{\comb{S}}{n} \comb{I}, \ldots \]
  We write $\comb{U}_n = \leftappiterate{\comb{B}\fpcC}{\comb{S}}{n}\comb{I}$ for the $n$-th term in this sequence. 
\end{defi}
The Scott sequence concurs with
the \boehm{} sequence of fpc's only for the first two elements, and then
splits off with different fpc's. 
But there is a second surprise. 
In showing that $\comb{U}_n$ is an fpc,
we find as a bonus the fpc-generating vector 
$\leftappiterate{\cxthole(\comb{S}\comb{S})}{\comb{S}}{n}\comb{I}$
(which does preserve the property of fpc's to be reducing).
%
%
\begin{thm}\label{thm:scott:sequence}
  Let $Y$ be a $k$-reducing fpc and $n \geq 0$. Then:
  \begin{enumerate}
  \item
      $\leftappiterate{\comb{B}Y}{\comb{S}}{n}\comb{I}$ is a (non-reducing) fpc;
  \item
    $\leftappiterate{Y(\comb{S}\comb{S})}{\comb{S}}{n}\comb{I}$ is a $(k+3n+7)$-reducing fpc.
  \end{enumerate}
\end{thm}
The proof of Theorem~\ref{thm:scott:sequence} is easy:
see the next example.

\begin{exa}
  Let $Y$ be a $k$-reducing fpc.
  Then:
  \begin{align*}
    \comb{B} Y \comb{S} \comb{S} \comb{S} \comb{I} x
    & \mhred Y (\comb{S} \comb{S}) \comb{S} \comb{I} x 
    \hredn{k} \comb{S} \comb{S} (Y (\comb{S} \comb{S})) \comb{S} \comb{I} x
    \\
    & \hredn{3} \comb{S} \comb{S} (Y (\comb{S} \comb{S}) \comb{S}) \comb{I} x
    \hredn{3} \comb{S} \comb{I} (Y (\comb{S} \comb{S}) \comb{S} \comb{I}) x
    \\
    & \hredn{3} \comb{I} x (Y (\comb{S} \comb{S}) \comb{S} \comb{I} x)
    \hredn{1} x (Y (\comb{S} \comb{S}) \comb{S} \comb{I} x)
  \end{align*}
  This shows that $\leftappiterate{\comb{B} Y}{\comb{S}}{3} \comb{I}$ is a non-reducing fpc,
  and at the same time that $Y (\comb{S}\comb{S}) \comb{S} \comb{I} x$ is reducing.
\end{exa}


\subsection{Generalized Generation Schemes}\label{sec:schemes}

The schemes mentioned in 
Theorem~\ref{thm:scott:sequence} 
for generating new fixed points from old, 
are by no means the only ones. 
There are in fact infinitely many of such schemes. 
They can be obtained analogously to the ones that we extracted
above from the equation $\comb{B}Y = \comb{B}Y\comb{S} = \lambda ab.(ab)^{\omega}$, 
or the equation $Mab = ab(Mab)$. 
We only treat the case for $n=3$: 
consider the equation $Nabc = abc(Nabc)$. 
Then every solution $N$ is again a `pre-fpc', 
namely $N\comb{I}\comb{I}$ is an fpc: $N\comb{I}\comb{I}x \conv \comb{I}\comb{I}x(N\comb{I}\comb{I}x) \conv x(N\comb{I}\comb{I}x)$. 
\begin{enumerate}
  \item
  \mbox{$Nabc = Y(abc)$,
  which yields $N = (\lambda yabc.y(abc)))Y$} $= (\lambda yabc.\comb{B}\comb{B}\comb{B}yabc)Y$.
  We obtain $N = \comb{B}\comb{B}\comb{B}Y$. 
  \item 
  $N = Y \xi$ with $\xi = \lambda nabc.abc(nabc)$, 
  yielding the fpc-generating vector $\cxthole \xi \comb{I}\comb{I}$. 
  \item $Nabc = abc(Nabc) = \comb{S}(ab)(Nab)c$.
  So we take $Nab = \comb{S}(ab)(Nab)$, which yields $Nab = Y(\comb{S}(ab)) = \comb{B}\comb{B}\comb{B}Y(\comb{B}\comb{S})ab$.
  So $N = \comb{B}\comb{B}\comb{B}Y(\comb{B}\comb{S})$, and thus we find the equation 
  $\comb{B}\comb{B}\comb{B}Y = \comb{B}\comb{B}\comb{B}Y(\comb{B}\comb{S})$,
  in analogy with the equation $\comb{B}Y = \comb{B}Y\comb{S}$ above. 
\end{enumerate}
Also this equation spawns lots of fpc's as well as fpc-generating vectors. 
Let's abbreviate $\comb{B}\comb{S}$ by $\comb{A}$. 
First one forms the sequence 
\begin{align*}
  \comb{B}\comb{B}\comb{B}Y,\; \comb{B}\comb{B}\comb{B}Y\comb{A},\; \comb{B}\comb{B}\comb{B}Y\comb{A}\comb{A},\; \comb{B}\comb{B}\comb{B}Y\comb{A}\comb{A}\comb{A},\ldots
\end{align*}
These terms all have the BT $\lambda abc.abc(abc)^{\omega}$. 
They are not yet fpc's , but only `pre-fpc's'. 
But after postfixing this time $\ldots \comb{I}\comb{I}$ we do again obtain a sequence of fpc's: 
\begin{align*}
  \comb{B}\comb{B}\comb{B}Y\comb{I}\comb{I},\; \comb{B}\comb{B}\comb{B}Y\comb{A}\comb{I}\comb{I},\; \comb{B}\comb{B}\comb{B}Y\comb{A}\comb{A}\comb{I}\comb{I}, \ldots
\end{align*}
Again the first two coincide with $\fpcC,\fpcT$, 
but the series deviates not only from the \boehm{} sequence 
but also from the Scott sequence above. 
As above, the proof that a term in this sequence is indeed an fpc, 
yields an fpc-generating vector. 
%
In this way we find as a new fpc-generating scheme
\[ Y \Rightarrow \leftappiterate{Y (\comb{A}\comb{A}\comb{A})}{\comb{A}}{n} \comb{I}\comb{I} \]
We can derive many more of these schemes by proceeding
with solving the general equation $Na_{1}a_{2}...a_{n}=a_{1}a_{2}...a_{n}(Na_{1}a_{2}...a_{n})$,
bearing in mind the following proposition.
\begin{prop}
  If $N$ is a term satisfying
  $N a_1 a_2 \ldots a_n = a_1 a_2 \ldots a_n (N a_1 a_2 \ldots a_n)$,
  then $\leftappiterate{N}{\comb{I}}{(n-1)}$ is an fpc.
\end{prop}

We finally mention an fpc-generating scheme with `dummy parameters':
\[ Y \Rightarrow Y Q P_{1} \ldots P_{n} \]
where $P_{1},\ldots,P_{n}$ are arbitrary (dummy) terms, 
and $Q = \mylam{yp_{1}...p_{n}x}{x(y p_1\ldots p_n x)}$.

\section{Clock Behaviour of Lambda Terms}\label{sec:clocked}

As we have seen, there is vast space of fpc's
and there are many ways to derive new fpc's.
The question is whether all these fpc's are indeed new.
So we have to prove that they are not $\beta$-convertible.

For the \bohm{} sequence we did this by an ad hoc argument
based on a syntactic invariant; 
and this method works fine to establish lots of non-equations
between the alleged `new' fpc's that we constructed above.
Still, the question remains whether there are 
not more `strategic' ways of proving such inequalities. 

In this section we propose a more strategic way 
to discriminate terms with respect to $\beta$-conversion.
The idea is to extract from a $\lambda$-term more than just its BT,
but also how the BT was formed; 
one could say, in what tempo, or in what rhythm.
A BT is formed from static pieces of information,
but these are rendered in a clock-wise fashion,
where the ticks of the internal clock are head reduction steps.

In the sequel we write $\annotate{k}{M}$ for
the term $M$ where the root is \emph{annotated with $k\in\nat$}.
Here, term formation binds stronger than annotation $\annotate{k}{}$.
For example $\annotate{k}{M N}$ stands for the term $\annotate{k}{(M N)}$
(that is, annotating the (non-displayed) application symbol in-between $M$ and $N$,
in contrast to $(\annotate{k}{M}) N$).
Moreover, for an annotated term $M$ we use $\deannotate{M}$
to denote the term obtained from $M$ by dropping all annotations (including annotations of subterms).
%
%
\begin{defi}[Clocked \bohm{} Trees]\label{def:cbohm}
  The \emph{clocked \bohm{} Tree $\cbohm{M}$} of a $\lambda$-term $M$
  is an annotated (potentially infinite) term coinductively defined as follows.
  If $M$ has no hnf, then define $\cbohm{M}$ as $\sink$.
  Otherwise,
  there is a head reduction $M \hredn{k} \mylam{x_1}{\ldots\mylam{x_n}{y M_1 \ldots M_m}}$ to hnf.
  Then we define 
  $\cbohm{M}$ as the term $\annotate{k}{\mylam{x_1}{\ldots\mylam{x_n}{y \cbohm{M_1} \ldots \cbohm{M_m}}}}$.
\end{defi}

The notions of subterms and positions (see Definition~\ref{def:pos}) 
carry over to annotated terms, in particular clocked \boehm{} Trees, in a straightforward way. 
The (non-clocked) \boehm{} Tree of a $\lambda$-term $M$
can be obtained by dropping the annotations:
$\sbohm(M) \defeq \deannotate{\cbohm{M}}$.


\begin{figure}[ht!]
  \begin{center}
  \begin{tikzpicture}[level distance=7mm,inner sep=1mm]
      \node  {$\treeap$} \annotatednode{$\treeap$}{2}
        child { node {$f$} }
        child { node {$\treeap$} \annotatednode{$\treeap$}{1}
          child { node {$f$} }
          child { node {$\treeap$} \annotatednode{$\treeap$}{1}
            child { node {$f$} }
            child { node {$\ddots$} \annotatednode{$\treeap$}{1}
            }
          }
        };
  \end{tikzpicture}
  \begin{tikzpicture}[level distance=7mm,inner sep=1mm]
      \node {$\treeap$} \annotatednode{$\treeap$}{2}
        child { node {$f$} }
        child { node {$\treeap$} \annotatednode{$\treeap$}{2}
          child { node {$f$} }
          child { node {$\treeap$} \annotatednode{$\treeap$}{2}
            child { node {$f$} }
            child { node {$\ddots$} \annotatednode{$\treeap$}{2}
            }
          }
        };
  \end{tikzpicture}
  \caption{\mbox{Clocked \bohm{} Trees of $\fpcC f$ and $\fpcT f$.}}
  \label{fig:boem:y0:y1}
  \end{center}
\end{figure}
Let us consider the fpc's $\fpcC$ of Curry and $\fpcT$ of Turing.
We have $\fpcC \equiv \mylam{f}{\omega_f\omega_f}$ 
where $\omega_f \equiv \mylam{x}{f(xx)}$, and
\begin{align*}
  \omega_f\omega_f \hredn{1} f (\omega_f\omega_f)
\end{align*}
Therefore we obtain $\cbohm{\fpcC f} = \annotate{2}{f \cbohm{\omega_f \omega_f}}$,
and $\cbohm{\omega_f \omega_f} = \annotate{1}{f \cbohm{\omega_f \omega_f}}$.

For $\fpcT \equiv \eta \eta$ where $\eta \equiv \mylam{x}{\mylam{f}{f (xxf)}}$ we get:
\begin{align*}
  \fpcT f \equiv \eta \eta f \hredn{2} f (\eta \eta f)
\end{align*}
Hence, $\cbohm{\fpcT f} = \annotate{2}{f \cbohm{\fpcT f}}$.
Figure~\ref{fig:boem:y0:y1} displays the 
clocked \bohm{} Trees of $\fpcC f$ (left) and $\fpcT f$ (right).

The following definition captures the well-known \boehm{} equality
of $\lambda$-terms.
\begin{defi}
  $\lambda$-terms $M$ and $N$ are \emph{$\sbohm$-equal},
  denoted by $M \treeequal{\sbohm} N$,
  if $\sbohm(M) \equiv \sbohm(N)$.
\end{defi}

If $M$ and $N$ are not $\sbohm$-equal, then $M \notconv N$.
Consequently, if for some $\lambda$-term~$F$, 
we have $\sbohm(M F) \not\equiv \sbohm(N F)$, then $M \notconv N$.

Below, we refine this approach by comparing
the clocked \bohm{} Trees $\cbohm{M}$ and $\cbohm{N}$
instead of the ordinary (non-clocked) \boehm{} Trees.
In general, $\cbohm{M} \not\equiv \cbohm{N}$
does not always imply that $M \notconv N$.
Nevertheless, for a large class of $\lambda$-terms, called `simple' below,
this implication will turn out to be true.

In the following definition, 
we lift relations over natural numbers to relations over clocked \boehm{} Trees.

\newcommand{\scbt}{T}
\newcommand{\cbt}{\sub{\scbt}}
\newcommand{\acbt}{\cbt{1}}
\newcommand{\bcbt}{\cbt{2}}

\begin{defi}
  Let $\acbt$ and $\bcbt$ be clocked \boehm{} Trees,
  and ${\mathrel{R}} \subseteq \nat \times \nat$.
  We define the following notations:
  \begin{enumerate}
    \item 
      For $\apos \in \pos{\acbt} \cap \pos{\bcbt}$
      we let $\acbt \relat{\mathrel{R}}{\apos} \bcbt$ denote 
      that either both $\subtrmat{\acbt}{\apos}$ and $\subtrmat{\bcbt}{\apos}$ are not annotated,
      or both are annotated and then
      $\subtrmat{\acbt}{\apos} \equiv \annotate{k_1}{\acbt'}$
      and $\subtrmat{\bcbt}{\apos} \equiv \annotate{k_2}{\bcbt'}$
      with $k_1 \mathrel{R} k_2$.

  \item 
    We write $\acbt \mathrel{R} \bcbt$
    if $\deannotate{\acbt} \equiv \deannotate{\bcbt}$
    and $\acbt \relat{\mathrel{R}}{\apos} \bcbt$ for every $\apos \in \pos{\acbt}$.
 
  \item 
    We write $\acbt \relev{\mathrel{R}} \bcbt$, and say that \emph{$\mathrel{R}$ holds eventually},
    if $\deannotate{\acbt} \equiv \deannotate{\bcbt}$ and
    there exists a depth level $\ell \in \nat$
    such that  $\acbt \relat{\mathrel{R}}{\apos} \bcbt$ for all positions $\apos \in \pos{\acbt}$ with $\length{p} \ge \ell$.

  \end{enumerate}

\end{defi}

%
%
%

%



\begin{defi}
For $\lambda$-terms $M$ and $N$ we say:
\begin{enumerate}
  \item \emph{$M$ improves $N$ globally} if $\cbohm{M} \rel{\le} \cbohm{N}$;
  \item \emph{$M$ improves $N$ eventually} if $\cbohm{M} \relev{\le} \cbohm{N}$; 
  \item \emph{$M$ matches $N$ eventually} if $\cbohm{M} \relev{=} \cbohm{N}$.
\end{enumerate}
\end{defi}

The following proposition states that
the ordering $\crel{>}$ on $\lambda$-terms defined by
$M \crel{>} N$ if and only if $\cbohm{M} \rel{\ge} \cbohm{N}$
is a `semi-model' of $\beta$-reduction~\cite{plotkin:91}. 
We leave this for future research.

\begin{prop}\label{prop:clocks}
  Clocks are accelerated under reduction, 
  that is, if $M {\mred} N$, then the reduct $N$ improves $M$ globally $(\cbohm{M} \rel{\ge} \cbohm{N})$.
  Dually, clocks slow down under expansion (the reverse of reduction).
\end{prop}

\begin{proof}
  We proceed by an elementary diagram construction.
  Whenever we have co-initial steps $M \hred M_1$ and $M \dred M_2$,
  then by orthogonal projection~\cite{terese:2003}
  there exist joining steps $M_1 \dred M'$ and $M_2 \hredeq M'$.
  Note that the head step $M \hred M_1$
  cannot be duplicated, only erased in case of an overlap.
  This leads to the elementary diagram displayed in Figure~\ref{fig:elementary:diagram}.
  \vspace{-3ex}
  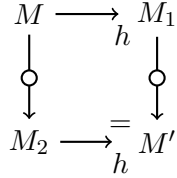
\begin{figure}[ht!]
  \begin{center}
  \begin{tikzpicture}[thick,node distance=17mm]
    \node (M) {$M$};
    \node (M1) [right of=M] {$M_1$};
    \node (M2) [below of=M] {$M_2$};
    \node (M') [below of=M1] {$M'$};
    \draw [->,shorten >= 1mm] (M) -- (M1) node [pos=.95,below] {$h$};
    \draw [->] (M) -- (M2); \fill [fill=white,draw=black] ($(M)!.5!(M2)$) circle (1mm);
    \draw [->,shorten >= 2mm] (M2) -- (M') node [pos=.9,below] {$h$} node [pos=.9,above] {$=$};
    \draw [->] (M1) -- (M'); \fill [fill=white,draw=black] ($(M1)!.5!(M')$) circle (1mm);
  \end{tikzpicture}
  \vspace{-2ex}
  \caption{Elementary diagram.}
  \vspace{-2ex}
  \label{fig:elementary:diagram}
  \end{center}
  \end{figure}

  We have ${\mred} \subseteq {\dred^*}$.
  By induction on the length of the rewrite sequence $\dred^*$
  it suffices to show that $M \dred N$ implies $\cbohm{M} \rel{\ge} \cbohm{N}$.
  Let $M \dred N$.
  If $M$ has no hnf, then the same holds for $N$, and hence $\cbohm{M} = \bot = \cbohm{N}$.
  Assume that there exists a head rewrite sequence 
  $M \hredn{k} H \equiv \mylam{x_1}{\ldots\mylam{x_n}{y M_1 \ldots M_m}}$ to hnf.
  We have $\cbohm{M} \equiv \annotate{k}{\mylam{x_1}{\ldots\mylam{x_n}{y \cbohm{M_1} \ldots \cbohm{M_m}}}}$.

  Using the elementary diagram above ($k$ times),
  we can project $M \dred N$ over $M \hredn{k} H$,
  and obtain
  $H \dred H'$, 
  $N \hredn{\ell} H' \equiv \mylam{x_1}{\ldots\mylam{x_n}{y M_1' \ldots M_m'}}$
  with $\ell \le k$.
  Then
  $\cbohm{N} \equiv \annotate{\ell}{\mylam{x_1}{\ldots\mylam{x_n}{y \cbohm{M_1'} \ldots \cbohm{M_m'}}}}$
  and $\ell \le k$.
  Since $H \dred H'$ and $H$ is in hnf,
  we get $M_i \dred M_i'$ for every $i = 1,\ldots,m$.
  We then can apply the same argument to $M_i \dred M_i'$
  and by coinduction (or induction on the depth), we obtain $\cbohm{M} \ge \cbohm{N}$.
\end{proof}

While $\cbohm{M} \not\equiv \cbohm{N}$ does not imply $M \notconv N$,
the following theorem allows us to use clocked \bohm{} Trees
for discriminating $\lambda$-terms:
%
%
\begin{thm}\label{thm:general}
  Let $M$ and $N$ be $\lambda$-terms.
  If $N$ cannot be improved 
  globally by any reduct of $M$, then $M \ne_\beta N$.
\end{thm}

\begin{proof}
  If $M =_\beta N$, then $M \mred M' \mredi N$
  for some $M'$ by confluence. Hence $\cbohm{M'} \rel{\le} \cbohm{N}$ by Proposition~\ref{prop:clocks}.
\end{proof}

\noindent
Note that for distinguishing $M$ and $N$ we can always consider
$\beta$-equivalent terms $M' \conv M$ 
and $N' \conv N$ instead.
For Theorem~\ref{thm:general} we have to show $\neg(\cbohm{M'} \rel{\le} \cbohm{N})$ for all reducts $M'$ of $M$.
This condition is in general difficult to prove.
However, the theorem is of use if one of the terms 
has a manageable set of reducts, and this term happens to have slower clocks.
A striking example will be given below in solving a question of Plotkin in Section~\ref{sec:plotkin}.

For a large class of $\lambda$-terms it turns out that clocks are invariant under reduction.
We call these terms `simple'.
\begin{defi}
  A redex $(\mylam{x}{M}){N}$ is called:
  \begin{enumerate}\setlength{\itemsep}{0ex}
    \item \emph{linear} if $x$ has at most one occurrence in $M$;
    \item \emph{call-by-value} if $N$ is a normal form; and
    \item \emph{simple} if it is linear or call-by-value.
  \end{enumerate}
\end{defi}

The definition of simple redexes generalizes the well-known notions
of call-by-value and linear redexes. 
Next, we define simple \emph{terms}.
Intuitively, we call a term $M$ `simple' if every reduction
admitted by $M$ only contracts simple redexes.
The following definition 
further generalizes this intuition by 
considering only reductions computing the \boehm{} Tree.
\begin{defi}[Simple terms]
  A $\lambda$-term $M$ is \emph{simple}
  if either $M$ has no hnf, 
  or the head reduction to hnf $M \mhred \mylam{x_1}{\ldots\mylam{x_n}{y M_1 \ldots M_m}}$
  contracts only simple redexes,
  and $M_1,\ldots,M_m$ are simple terms.
\end{defi}
\noindent
Note that this definition is essentially coinductive:
the set of simple terms is the largest set $X$ such that 
if $M \in X$ then either $M$ has no hnf or the reduction to hnf
$M \mhred \mylam{x_1}{\ldots\mylam{x_n}{y M_1 \ldots M_m}}$
contracts only simple redexes and $M_1,\ldots,M_n \in X$ again.

All the fpc's in this paper are either simple or have simple reducts.
The clock of simple $\lambda$-terms is invariant under reduction,
that is, reduction of a simple term affects only finitely many annotations 
in the clocked \bohm{} Tree.
For example, by reducing a term we can always make the clock values
in a finite prefix equal to $0$.

\begin{prop}\label{prop:simple}
  Let $N$ be a reduct of a simple term $M$. Then $N$ matches $M$ eventually
  $(\cbohm{M} \relev{=} \cbohm{N})$.
\end{prop}
\begin{proof}
  The proof is a straightforward extension of the proof of~Proposition~\ref{prop:clocks}
  with the observation that for simple terms $M$,
  rewriting
  $M \hredn{k} H \equiv \mylam{x_1}{\ldots\mylam{x_n}{y M_1 \ldots M_m}}$ to hnf
  does not duplicate redexes.
  For simple terms $M$, the elementary diagrams
  are of the form displayed in Figure~\ref{fig:elementary:diagram:simple}.
  We use $s \stackrel{\varnothing}{\longrightarrow} t$ to denote empty steps, that is, $s \equiv t$.
  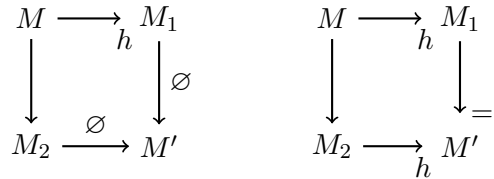
\begin{figure}[h!]
  \begin{center}
  \begin{tikzpicture}[thick,node distance=17mm]
    \node (M) {$M$};
    \node (M1) [right of=M] {$M_1$};
    \node (M2) [below of=M] {$M_2$};
    \node (M') [below of=M1] {$M'$};
    \draw [->,shorten >= 1mm] (M) -- (M1) node [pos=.95,below] {$h$};
    \draw [->] (M) -- (M2); 
    \draw [->] (M2) -- (M') node [midway,above] {$\varnothing$};
    \draw [->] (M1) -- (M') node [midway,right] {$\varnothing$};

    \begin{scope}[xshift=4cm]
    \node (M) {$M$};
    \node (M1) [right of=M] {$M_1$};
    \node (M2) [below of=M] {$M_2$};
    \node (M') [below of=M1] {$M'$};
    \draw [->,shorten >= 1mm] (M) -- (M1) node [pos=.95,below] {$h$};
    \draw [->] (M) -- (M2); 
    \draw [->,shorten >= 1mm] (M2) -- (M') node [pos=.9,below] {$h$};
    \draw [->,shorten >= 1.5mm] (M1) -- (M') node [pos=.9,right] {$=$};
    \end{scope}
  \end{tikzpicture}
  \vspace{-2ex}
  \caption{Elementary diagrams for simple $M$.}
  \vspace{-2ex}
  \label{fig:elementary:diagram:simple}
  \end{center}
  \end{figure}
  
  We briefly explain the two elementary diagrams.
  They depict the possible scenarios for joining co-initial steps $M \hred M_1$ and $M \to M_2$, that is,
  a head reduction step set out against an arbitrary reduction step.
  If $M$ is a simple term,
  then either 
  \begin{enumerate}
    \item\label{item:cancel}
    both are the same step, and can be joined with empty steps $\stackrel{\varnothing}{\longrightarrow}$, or
    \item\label{item:proj}
    they can be joined by $M_2 \hred M' \redi^= M_1$
    since a head step in a simple term with hnf cannot duplicate a redex (but deletion is possible).
  \end{enumerate}
  We show the following implication by induction on $n \in \nat$ (employing the diagrams in Figure~\ref{fig:elementary:diagram:simple}):
  \begin{align}
    M_1 \hredi M \to^n N \quad\implies\quad M_1 \to^{\le n} O \hredi N \;\vee\; M_1 \to^{< n} N
    \label{simple:join}
  \end{align}
  For $n = 0$, there is nothing to be shown.
  Let $n > 0$, and consider $M_1 \hredi M \to M_2 \to^{n-1} N$.
  Then either \eqref{item:cancel} $M_1 \equiv M_2$ and consequently $M_1 \to^{n-1} N$,
  or \eqref{item:proj} $M_1 \to^= M' \hredi M_2$
  and by induction hypothesis $M' \to^{\le n-1} O \hredi M_2$ or $M' \to^{< n-1} N$,
  yielding $M_1 \to^{\le n} O \hredi M_2$ or $M_1 \to^{< n} N$, respectively.

  For $\lambda$-terms $M$ and integers $n\in\nat$,
  we define finite approximations $\nbohm{n}{M}$ of the clocked \bohm{} Tree $\cbohm{M}$ of $M$.
  We let $\nbohm{0}{M} = M$, and for $n > 0$ define 
  \[\nbohm{n}{M} = \annotate{k}{\mylam{x_1}{\ldots\mylam{x_n}{y \nbohm{n-1}{M_1} \ldots \nbohm{n-1}{M_m}}}}\]
  if $M \hredn{k} \mylam{x_1}{\ldots\mylam{x_n}{y M_1 \ldots M_m}}$ to hnf,
  and  
  $\nbohm{n}{M} = \sink$ if $M$ has no hnf.
  Note that $\cbohm{M} = \lim_{n\to\infty} \nbohm{n}{M}$.

  Let $M,N$ be terms, $M$ simple and $M \to^n N$.
  Then we claim  
  \begin{align}
    \deannotate{\nbohm{1}{M}} \to^{\le n'} \deannotate{\nbohm{1}{N}}
    \tag*{$(*)$}
  \end{align}
  where $n' = n$ if $\nbohm{1}{M}$ and $\nbohm{1}{N}$ have the same annotation,
  and $n' = n-1$~otherwise.
  If $M$ has no hnf, then $N$ has no hnf, and $\nbohm{1}{M} = \nbohm{1}{N} = \bot$.
  Thus assume that~$M$ admits a head reduction $M \hredn{k} M' \equiv \mylam{x_1}{\ldots\mylam{x_n}{y M_1 \ldots M_m}}$ 
  to hnf for some $k\in\nat$.
  Then by induction using \eqref{simple:join}
  we obtain that either 
  \begin{align*}
    \text{(a) }N \hredn{k} N' \redi^{\le n} M' &&
    \text{or} &&
    \text{(b) }N \hredn{<k} N' \redi^{< n} M'
  \end{align*}
  for some $N' \equiv \mylam{x_1}{\ldots\mylam{x_n}{y N_1 \ldots N_m}}$.
  This proves the claim.

  Note that every reduction $\nbohm{m}{M} \to^{\le n} \nbohm{m}{N}$
  is contained within the subterms $M'$ of $\nbohm{m}{M}$ that are 
  left unreduced by the clause $\nbohm{0}{M'} = M'$;
  everything outside these subterms is in normal form.
  Using the previous observation and $(*)$,
  we obtain by induction that for all $m \in \nat$:
  $\deannotate{\nbohm{m}{M}} \to^{\le n-d_m} \deannotate{\nbohm{m}{N}}$
  where $d_m$ is the number of positions $p$ where 
  the annotation of $\subtrmat{\nbohm{m}{M}}{p}$ differs from the annotation of $\subtrmat{\nbohm{m}{N}}{p}$. 
  As a consequence,
  the annotations of $\cbohm{M}$ and $\cbohm{N}$
  differ at most at $n$ positions.
  Hence $\cbohm{M} \relev{=} \cbohm{N}$.%
\end{proof}

Reduction accelerates clocks, and 
for simple terms the clock is invariant under reduction, see Proposition~\ref{prop:simple}.
Hence if a term $M$ has a simple reduct $N$, then $N$ has the fastest clock reachable from $M$ 
modulo a finite prefix. This justifies the following convention.
\begin{cOnv}
  The \emph{(minimal) clock} of a $\lambda$-term $M$ with a simple reduct $N$
  is $\cbohm{N}$, the clocked BT of $N$.
\end{cOnv}

For simple terms we obtain the following theorem:
%
%
\begin{thm}\label{thm:simple}
  Let $M$ and $N$ be $\lambda$-terms such that $M$ is simple.
  If $M$ does not improve eventually on $N$, then $M \ne_\beta N$.
\end{thm}
\begin{proof}
  Assume $M =_\beta N$. Then
  $M \mred M' \mredi N$ for some $M'$.
  We have $\cbohm{M} \relev{=} \cbohm{M'}$ by Proposition~\ref{prop:simple}, and
  $\cbohm{M'} \rel{\le} \cbohm{N}$ by Proposition~\ref{prop:clocks}.
  Hence we obtain \alert{$\cbohm{M} \relev{\le} \cbohm{N}$}.
\end{proof}
Theorem~\ref{thm:simple} significantly reduces 
the proof obligation in comparison to Theorem~\ref{thm:general}.
We only consider the clocked BT's of $M$ and $N$,
instead of all reducts $M'$ of $M$.

Note that Theorem~\ref{thm:simple} can also be employed for discriminating non-simple $\lambda$-terms
if one of the terms has a simple reduct.
For the case that both $M$ and $N$ are simple, there is no need to look for reducts
since the eventual clocks are invariant under reduction, see Proposition~\ref{prop:simple}:
\begin{cor}\label{cor:simple:simple}
  If simple terms $M$, $N$ do not match eventually,
  then they are not $\beta$-convertible, that is, $M \ne_\beta N$.
\end{cor}
\begin{proof}
  Assume $M = N$, then $M \mred O \mredi N$ for a common reduct $O$.
  Then $\cbohm{M} \relev{=} \cbohm{O} \relev{=} \cbohm{N}$ by Proposition~\ref{prop:simple}. 
  Hence $\cbohm{M} \relev{=} \cbohm{N}$, that is, $M$ and $N$ match eventually.
\end{proof}

\begin{rem}\label{rem:prefix}
  The reason for the qualifier `eventually' in the notions above,
  in other words, working modulo a finite prefix of the BT, 
  is that by some preliminary reduction we can always make the clock 
  values in any finite prefix equal to $0$. 
  So we are interested exclusively in the `tail behaviour', 
  or the behaviour `at infinity', and not in the initial 
  behaviour of the development to the BT. 
  
  To give a concrete example:
  $\fpc{0}$ and $\fpc{1}$, the fpc's of Curry and Turing, can be reduced to reducts 
  $M_0$, $M_1$ respectively, that have an initial segment of arbitrary length $n$
  of their BT's with clock labels $0$ (just reduce first to $\mylam{f}{f^n(\ldots)}$).
  However, the infinite remainders of their BT's, their tails as it were, 
  will reveal the difference in clock values, witnessing the fact that $\fpc{0}$ 
  eventually improves $\fpc{1}$. And this situation is stable under reduction;
  indeed, for any two reducts $M_0$, $M_1$ as above, the first eventually improves the second.
\end{rem}

\begin{rem}
  Take $\lambda$-terms $M$, $N$ with finite BT's.
  Then the clock comparison of $M$ and $N$ amounts to the comparison of their non-clocked BT's.
  In case their BT's are equal and \mbox{$\bot$-free}, then $M \conv N$.
  If their BT's coincide but are not \mbox{$\bot$-free}, then we can fine-tune the analysis
  by comparing their clocked \levy{} or \ber{} Trees
  (a $\bot$ in a BT may give rise to an infinite \ber{} Tree), see Section~\ref{sec:levy}.
\end{rem}

\begin{exa}\label{ex:boehm:seq}
  We compute the clocks of  
  $\fpc{n} x$ with $\fpc{n}$ the $n$-th term of the \boehm{} sequence.
  The clocks of $\fpc{0} x$ and $\fpc{1} x$ have been computed before, 
  see Figure~\ref{fig:boem:y0:y1}.
  We now compute the clock of $\fpc{n} x$ for $n \geq 2$:
  \begin{align*}
    \fpc{n}
    & \equiv
    \leftappiterate{\eta \eta}{\delta}{n-1} x \\
    & \hredn{2} \leftappiterate{\delta (\eta \eta \delta)}{\delta}{n-2} x
    \\
    & \hredn{2(n-2)} \delta (\leftappiterate{\eta \eta}{\delta}{n-1}) x
    \\
    & \hredn{2} x (\leftappiterate{\eta \eta}{\delta}{n-1} x) 
  \end{align*}
  Notice that none of these steps duplicate a redex, 
  hence $\fpc{n}$ is a simple term. 
  We find $\cbohm{\fpc{n} x} 
   = \annotate{2n}{(x \cbohm{\fpc{n} x})}$.
  Hence, for $n \geq 2$ the clock of $\fpc{n}$ is $2n$.
  Consequently, all terms in the \boehm{} sequence have distinct clocks.
  Thus we have an alternative proof of Theorem~\ref{thm:boehm:seq}: 
  the \boehm{} sequence contains no duplicates.
\end{exa}
%

\begin{exa}\label{ex:scott:seq}
  Let $n \geq 2$.
  We compute the clocks of the fpc's 
  $\comb{U}_n \defeq \leftappiterate{\comb{B}\fpcC}{\comb{S}}{n} \comb{I}$
  of the Scott sequence. 
  We first reduce $\comb{U}_n x$ to a simple term:
  \begin{align*}
    \comb{U}_n x 
    & \mred \leftappiterate{\fpcC (\comb{S} \comb{S})}{\comb{S}}{(n-2)} \comb{I} x \\
    & \to \leftappiterate{\omega_{\comb{S}\comb{S}}\,\omega_{\comb{S}\comb{S}}}{\comb{S}}{(n-2)} \comb{I} x \\
    & \mred \leftappiterate{\nf{\omega_{\comb{S}\comb{S}}}\,\nf{\omega_{\comb{S}\comb{S}}}}{\comb{S}}{(n-2)} \comb{I} x
  \end{align*}
  where $\nf{\omega_{\comb{S}\comb{S}}} \equiv \mylam{abc}{bc(aabc)}$.
  We abbreviate $\theta = \nf{\omega_{\comb{S}\comb{S}}}$.
  Then we compute the clocks for $n = 2$, $n = 3$, and $n > 3$:
  \begin{align*}
    \theta\theta \comb{I} x 
    & \hredn{3} \comb{I} x (\theta\theta \comb{I} x) 
    \hredn{1} x (\theta\theta \comb{I} x) 
    \\
    \theta\theta \comb{S} \comb{I} x 
    & \hredn{3} {\comb{S} \comb{I}(\theta\theta \comb{S} \comb{I})} x
    \hredn{4} {x (\theta\theta \comb{S} \comb{I} x)} 
    \\
    \leftappiterate{\theta\theta}{\comb{S}}{(n-2)} \comb{I} x
    &\hredn{3} \leftappiterate{ \comb{S} \comb{S} (\theta\theta \comb{S} \comb{S})}{\comb{S}}{(n-4)} \comb{I} x
    \\
    &\hredn{3(n-4)} \comb{S} \comb{S} (\leftappiterate{ \theta \theta \comb{S} \comb{S} }{\comb{S}}{(n-4)}) \comb{I} x
    \\
    &\hredn{3} \comb{S} \comb{I} (\leftappiterate{ \theta \theta }{\comb{S}}{(n-2)} \comb{I}) x
    \\
    &\hredn{4} x (\leftappiterate{ \theta \theta }{\comb{S}}{(n-2)} \comb{I} x)
  \end{align*}
  respectively.
  For all three cases, we find:\\ 
  $\cbohm{\leftappiterate{\theta\theta}{\comb{S}}{(n-2)} \comb{I} x} 
   = \annotate{3n-2}{(x \cbohm{\leftappiterate{ \theta \theta }{\comb{S}}{(n-2)} \comb{I} x})}$.
\end{exa}
\noindent
Using Theorem~\ref{thm:simple} we infer from Example~\ref{ex:scott:seq}
and Figure~\ref{fig:boem:y0:y1} (recall that $\comb{U}_0 \conv \fpcC$ and $\comb{U}_1 \conv \fpcT$):
\begin{cor}
  The Scott sequence contains no duplicates.  
\end{cor}

\section{An Answer to a Question of Plotkin}\label{sec:plotkin}

Plotkin~\cite{plot:2007} asked:
\textit{Is there a fixed point combinator $\afpc$ such that}
\begin{align}
  A_\afpc \equiv \afpc(\mylam{z}{fzz}) \conv \afpc(\mylam{x}{\afpc(\mylam{y}{fxy})}) \equiv B_\afpc
  \label{eq:plotkin}
\end{align}
or in other notation: 
\begin{align*}
  \mymu{z}{fzz} \conv \mymu{x}{\mymu{y}{fxy}}\;,
\end{align*}
with the usual definition $\mymu{x}{M(x)} = Y(\mylam{x}{M(x)})$.
The terms $A_\afpc$ and $B_\afpc$ have the same \bohm{} Trees, namely
the solution of $T = f T T$.

Plotkin's question is pertinent to the question whether
absolutely unorderable models of the $\lambda$-calculus exist,
see Selinger's work~\cite{seli:1996,seli:1997,seli:2003}.
The negative solution of Plotkin's question 
blocks an appeal to Lemma 3.6 in~\cite{seli:2003}
to show that the generalized Mal'cev equations 
for all $n$ are inconsistent with the $\lambda\beta$-calculus. 
(For $n = 0,1$ this is established, but the general case is a difficult open problem.)

We show that there is no fpc $\afpc$ satisfying~\eqref{eq:plotkin}.
We begin with an example.

\begin{exa}
  We consider Turing's fpc $\fpcT$,
  and compute the clocked BT's of $A_{\fpcT}$ and $B_{\fpcT}$.
  Recall that $\fpcT \equiv \eta\eta$ with $\eta \equiv \mylam{xf}{f(xxf)}$.
  \begin{align*}
    A_{\fpcT} 
    &\equiv \eta\eta(\mylam{z}{fzz})
    \hredn{2} (\mylam{z}{fzz}) A_{\fpcT} 
    \hredn{1} f A_{\fpcT} A_{\fpcT} 
    \\
    B_{\fpcT} 
    &\equiv \eta\eta(\mylam{x}{\eta\eta(\mylam{y}{fxy})})
    \hredn{2} 
    (\mylam{x}{\eta\eta(\mylam{y}{fxy})}) B_{\fpcT} 
    \\
    &\hredn{1} 
    \eta\eta(\mylam{y}{f B_{\fpcT} y})
    \hredn{2}
    (\mylam{y}{f B_{\fpcT} y})(\eta\eta(\mylam{y}{f B_{\fpcT} y}))
    \\
    &\hredn{1}
    f B_{\fpcT} (\eta\eta(\mylam{y}{f B_{\fpcT} y}))
  \end{align*}
  Note that for $B_{\fpcT}$ developing the left branch takes six steps,
  whereas the right branch only needs three.
  We remark that this is not sufficient to conclude that $A_{\fpcT} \nconv B_{\fpcT}$
  since $A_{Y}$ and $B_{Y}$ are non-simple even for simple fpc's $Y$. 
  The clocked BT's for $A_{\fpcT}$ and $B_{\fpcT}$ are depicted in Figure~\ref{fig:plotkin}
  using hnf-notation, see~\cite{bare:1984} or~\cite{bare:klop:2009}.

  \begin{figure}[ht!]
  \begin{center}
    \begin{tikzpicture}[level distance=7mm,inner sep=0.5mm,
                        level 1/.style={sibling distance=18mm},
                        level 2/.style={sibling distance=9mm},
                        level 3/.style={sibling distance=4.5mm},
                        level 4/.style={sibling distance=2.25mm}
                       ]
      \node  {$f$} \annotatednode{$f$}{3}
        child { node {$f$} \annotatednode{$f$}{3}
          child { node {$f$} \annotatednode{$f$}{3}
            child { node {$\ldots$} }
            child { node {$\ldots$} }
          }
          child { node {$f$} \annotatednode{$f$}{3}
            child { node {$\ldots$} }
            child { node {$\ldots$} }
          }
        }
        child { node {$f$} \annotatednode{$f$}{3}
          child { node {$f$} \annotatednode{$f$}{3}
            child { node {$\ldots$} }
            child { node {$\ldots$} }
          }
          child { node {$f$} \annotatednode{$f$}{3}
            child { node {$\ldots$} }
            child { node {$\ldots$} }
          }
        };
    \end{tikzpicture}\quad
    \begin{tikzpicture}[level distance=7mm,inner sep=0.5mm,
                        level 1/.style={sibling distance=18mm},
                        level 2/.style={sibling distance=9mm},
                        level 3/.style={sibling distance=4.5mm},
                        level 4/.style={sibling distance=2.25mm}
                       ]
      \node  {$f$} \annotatednode{$f$}{6}
        child { node {$f$} \annotatednode{$f$}{6}
          child { node {$f$} \annotatednode{$f$}{6}
            child { node {$\ldots$} }
            child { node {$\ldots$} }
          }
          child { node {$f$} \annotatednode{$f$}{3}
            child { node {$\ldots$} }
            child { node {$\ldots$} }
          }
        }
        child { node {$f$} \annotatednode{$f$}{3}
          child { node {$f$} \annotatednode{$f$}{6}
            child { node {$\ldots$} }
            child { node {$\ldots$} }
          }
          child { node {$f$} \annotatednode{$f$}{3}
            child { node {$\ldots$} }
            child { node {$\ldots$} }
          }
        };
    \end{tikzpicture}
  \caption{Clocked BT's for $A_{\fpcT}$ and $B_{\fpcT}$, in hnf-notation.}
  \label{fig:plotkin}
  \end{center}
  \vspace{-2ex}
  \end{figure}
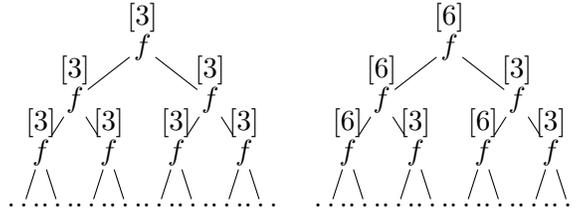
\end{exa}

To prove $A_Y \nconv B_Y$ for every fpc $Y$, 
we construct a term $B_Y'$ convertible with $B_Y$
such that no reduct of $A_Y$ improves globally on $B'_Y$.
Then by Theorem~\ref{thm:general} we have $A_Y \nconv B'_Y$,
and thus $A_Y \nconv B_Y$.
We have $B_Y \conv B'_Y$ where $B'_Y$ is defined by:
\begin{align*}
      B'_Y \equiv Y(\mylam{x}{fx(fx(Y(\mylam{y}{fxy})))}) 
\end{align*}
While the annotations of the clocked \bohm{} Tree $\cbohm{B_Y'}$ of $B_Y'$ of course depend on $Y$,
we can derive partial knowledge as follows.
We have $B'_Y \equiv \subst{(Yx)}{x}{M}$, where $M \equiv \mylam{x}{fx(fx(Y(\mylam{y}{fxy})))}$.
Then we have for every $Y' \conv Yx$:
\begin{align*}
  \subst{Y'}{x}{M} &\hredn{*} \subst{(x\, Y'')}{x}{M} \equiv M (\subst{Y''}{x}{M})\\
    &\hred f \;(\subst{Y''}{x}{M})\; \underline{(f \; (\subst{Y''}{x}{M}) \; (Y(\mylam{y}{f\,(\subst{Y''}{x}{M})\,y}))} )
\end{align*}
for some $Y'' \conv Y x$.
The underlined term in head normal form is obtained without reducing,
and hence will be annotated with $\annotate{0}{}$ in the clocked \bohm{} tree of $\subst{Y'}{x}{M}$.
The same argument applies for $\subst{Y''}{x}{M}$, and so on. 
Hence the second argument of every $f$ on the leftmost spine of $\cbohm{\subst{Y'}{x}{M}}$
has annotation $\annotate{0}{}$, as in Figure~\ref{fig:bty'}, which depicts 
the clocked \bohm{} tree $\cbohm{B_Y'}$ of $B_Y'$.
Therefore we obtain:
\begin{lem}\label{lem:BY'}
  Every annotation at a position of the form $(12)^{\ast}2$ 
  in $\cbohm{B_Y'}$ is $\annotate{0}{}$.
  \qed
\end{lem}
\noindent
Note that in the statement of the above lemma,
the positions $(12)^{n}2$ refer to trees in applicative notation, 
which is equivalent to $1^{n}2$ in the hnf-notation, 
as seen in Figure~\ref{fig:bty'}.

\begin{figure}[ht!]
\begin{center}
  \begin{tikzpicture}[level distance=7mm,inner sep=0.5mm,
                      level 1/.style={sibling distance=36mm},
                      level 2/.style={sibling distance=18mm},
                      level 3/.style={sibling distance=9mm},
                      level 4/.style={sibling distance=4.5mm}
                     ]
    \node  {$f$} \tinyannotatednode{$f$}{?}
      child { node {$f$} \tinyannotatednode{$f$}{?}
        child { node {$f$} \tinyannotatednode{$f$}{?}
          child { node {$f$} \tinyannotatednode{$f$}{?}
            child { node {$\ldots$} }
            child { node {$\ldots$} }
          }
          child { node {$f$} \annotatednode{$f$}{0}
            child { node {$\ldots$} }
            child { node {$\ldots$} }
          }
        }
        child { node {$f$} \annotatednode{$f$}{0}
          child { node {$f$} \tinyannotatednode{$f$}{?}
            child { node {$\ldots$} }
            child { node {$\ldots$} }
          }
          child { node {$f$} \tinyannotatednode{$f$}{?}
            child { node {$\ldots$} }
            child { node {$\ldots$} }
          }
        }
      }
      child { node {$f$} \annotatednode{$f$}{0}
        child { node {$f$} \tinyannotatednode{$f$}{?}
          child { node {$f$} \tinyannotatednode{$f$}{?}
            child { node {$\ldots$} }
            child { node {$\ldots$} }
          }
          child { node {$f$} \tinyannotatednode{$f$}{?}
            child { node {$\ldots$} }
            child { node {$\ldots$} }
          }
        }
        child { node {$f$} \tinyannotatednode{$f$}{?}
          child { node {$f$} \tinyannotatednode{$f$}{?}
            child { node {$\ldots$} }
            child { node {$\ldots$} }
          }
          child { node {$f$} \tinyannotatednode{$f$}{?}
            child { node {$\ldots$} }
            child { node {$\ldots$} }
          }
        }
      };
  \end{tikzpicture}\quad
\caption{Clocked \bohm{} Tree for $B_Y'$ in hnf-notation.}
\label{fig:bty'}
\end{center}
\vspace{-2ex}
\end{figure}
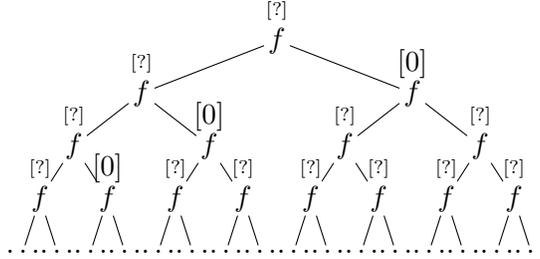

We continue by proving that no reduct of $A_Y$ improves globally on $B_Y'$.
More precisely, we show that for every reduct $A_Y'$ of $A_Y$ there exists
a position $\apos$ of the form $(12)^*2$ such that the annotation of $\cbohm{A_Y'}$ 
at $\apos$ is nonzero.
%
%

\begin{defi}
  A reduct $M$ of $A_Y \equiv Y(\mylam{z}{fzz})$ is called \emph{balanced} if
  $s \equiv t$ for every subterm $f\, s\, t$ in $M$
  where $f$ is a descendant of the displayed $f$ (see also Equation~\eqref{eq:plotkin}).
\end{defi}

Note that, in case $Y$ is not closed,
the term may contain occurrences~$f$.
In the definition of `balanced' we are not interested in those occurrences of $f$,
but only the residuals of $f$ displayed in $Y(\mylam{z}{fzz})$, that is,
the $f$ in $\mylam{z}{fzz}$.
Let us label this $f$ as $f^\star$, i.e. $Y(\mylam{z}{f^\star zz})$, such that $f^\star$
does not occur in $Y$.
Then by `residuals of the displayed $f$' we mean the free occurrences of $f^\star$
in all reducts of $Y(\mylam{z}{f^\star zz})$.

The following lemma states that every reduct of $A_Y$ can be balanced:

\begin{lem}\label{lem:symm}
  Let $A_Y \equiv Y (\mylam{z}{fzz}) \mred M$.
  There exists a balanced $N$ such that $M \mred N$.
\end{lem}

\begin{proof}
  Let $\arsdav$ denote 
  complete developments of the set of all redex occurrences (also known as Gross--Knuth steps).
  We consider the sequence $A_Y \equiv N_1 \arsdav N_2 \arsdav \ldots$.
  By cofinality of $\arsdav$~\cite{bare:1984,terese:2003}
  there exists $i \in \nat$ such that $M \mred N_i$.
  We define $N \equiv N_i$.

  It remains to be shown that the term $N$ is balanced.
  This follows from the fact that $A_Y$ is balanced,
  and that balancedness is preserved under~$\arsdav$.
  The latter can be seen as follows: 
  consider a term $\cxtfill{\acxt}{f\, s\, s}$.
  Obviously both displayed occurrences of $s$ contain the same redexes, 
  and the same variables bound from above $f\,s\,s$ in $\acxt$.
  Hence all descendants of both occurrences of $s$ after $\arsdav$ will again be identical.
\end{proof}

\begin{lem}\label{lem:AY}
  Let $A_Y \mred A_Y'$. There exists a position $\apos$ of the form $(12)^*2$ 
  such that $\cbohm{A_Y'}$ has a nonzero annotation at position $p$.
\end{lem}
\begin{proof}
  Let $A_Y''$ be a reduct of $A_Y'$ such that $A_Y''$ is in hnf.
  By Lemma~\ref{lem:symm} there exists a balanced reduct $A_Y'''$ of $A_Y''$.
  Note that also $A_Y'''$ is in hnf and its head symbol is $f$. 
  
  Let $\bpos$ be the shortest position of the form $(12)^*$ such that 
  $\subtrmat{A_Y'''}{\posconcat{\bpos}{12}}$ is not in hnf.
  Then $\subtrmat{A_Y'''}{\bpos}$ is in hnf, and thus $\trmat{A_Y'''}{\posconcat{\bpos}{11}} = f$.
  Because $A_Y'''$ is balanced we have that $\subtrmat{A_Y'''}{\posconcat{\bpos}{12}} \equiv \subtrmat{A_Y'''}{\posconcat{\bpos}{2}}$,
  and so $\subtrmat{A_Y'''}{\posconcat{\bpos}{2}}$ is not in hnf, too.
  Hence the annotation of $\cbohm{A_Y'''}$ at the position $\apos = \bpos 2$ is nonzero.
  The claim follows since $\cbohm{A_Y'''} \leq \cbohm{A_Y'}$.
\end{proof}

The combination of Lemmas~\ref{lem:BY'} and~\ref{lem:AY} implies that
there exists no reduct of $A_Y$ that improves globally on $B_Y'$.
Therefore by Theorem~\ref{thm:general} we obtain:

\begin{prop}
  There exists no fpc $Y$ such that
  $Y(\mylam{z}{fzz}) \conv Y(\mylam{x}{Y(\mylam{y}{fxy})})$.
  \qed
\end{prop}
%

\section{Clocked \boehm{} Trees and Periodic Terms}
\label{sec:periodic}

In the last two sections
we saw that the concept of clocked \boehm{} Trees
can be successfully used to discriminate between
$\lambda$-terms with equal \boehm{} Trees 
or to analyze the existence of terms satisfying certain equations.
In this section we will address the following question:
to which classes of $\lambda$-terms is the clocks method
most readily applicable?

A common feature of terms studied
in Sections \ref{sec:clocked} and \ref{sec:plotkin}
is that their \boehm{} Trees are `periodic':
they have inherent cyclical content that
allows us to analyze how their clocks are
developed towards infinity.  For example,
$Y$ is an fpc if the term $Y_x \equiv Yx$ satisfies
\begin{equation} \label{eq:fpc}
  Y_x = x (Y_x)
\end{equation}
That is, the subterm of $Y_x$ occurring at position
$2$ is equal to the term $Y_x$ itself.
If we use \boehm{} Tree equality instead of
$\beta$-equality then the equation \eqref{eq:fpc} would say
that $Y$ is a wfpc, or a looping combinator:
\[
\bt{Y_x} = \bt{x (Y_x)}
\]
The next
definition attempts to capture such situations
in a general fashion.

In what follows, $M_\sigma$ denotes the
term at position $\sigma$ in the \boehm{} Tree of $M$,
with $M_\sigma = \Omega = \comb{Y}_0 \comb{I}$ in the event that $\sigma \notin \pos{\bt{M}}$.
By a \emph{B\"ohm Tree context}
$C\cxthole_1\cdots\cxthole_n$ we mean a finite
B\"ohm Tree in which the leaves, in addition
to being variables or $\bot$, are allowed to be either
of the holes $\cxthole_i$.  (Equivalently, it is
a $\lambda\bot$-context which is redex-free.)

\begin{defi}
  Let $M \in \lterm$, and $\posemp \neq \sigma \in \pos{\bt{M}}$.
  \begin{enumerate}
  \item $M$ is \emph{periodic at $\sigma$} if $M =_\beta M_\sigma$.
  \item $M$ is \emph{weakly periodic at $\sigma$} if
    $M =_{\mathsf{BT}} M_\sigma$.
  \item $M$ is \emph{(weakly) locally periodic} if
    for every infinite path $x = \seqof{x_0, x_1, \dots}$
    through the B\"ohm Tree of $M$ we have that
    $M$ is (weakly) periodic at $x_n$ for some $n$.
  \item $M$ is \emph{(weakly) globally periodic} if
    there is a B\"ohm Tree context
    $C\cxthole_1\cdots\cxthole_n$ such that
    $M \equiv C[M_1]\cdots[M_n]$ and $M =_\beta M_i$
    ($M =_{\mathsf{BT}} M_i$).
  \end{enumerate}
\end{defi}

\begin{exa}
  \mbox{}
  \begin{enumerate}
    \item 
      If $Y$ is an fpc, then $Yx$ is periodic at $2$.
  
    \item 
      If $Y$ is a weak fpc, then $Yx$ is weakly periodic at $2$.
  
    \item 
      Using the fixed point theorem, let $M$ be such that $M = \lambda z. z M M$.  
      Then $M$ is periodic at $02$, $012$.
      Since any infinite path through $\bt{M}$ passes through one of these positions, 
      $M$ is locally periodic.
      In fact, $M$ is globally periodic with the context
      $C\cxthole_1\cxthole_2 \equiv \lambda z. z\cxthole_1\cxthole_2$.
  \end{enumerate}
  Generally, any term which is defined by an application of the
  fixed point theorem is periodic by construction.  If a looping
  combinator is used instead of an fpc, then the resulting term
  will be weakly periodic.  This suggests that (w)fpc's play a
  central role in the study of general periodic terms.
\end{exa}

We prove the following elementary fact:
\begin{prop}
  Let $M \in \lterm$.   Then
  \[
  M \text{ is (weakly) locally periodic } \iff
  M \text{ is (weakly) globally periodic } \]
\end{prop}
\begin{proof}
  The statement is proved simultaneously for weak and strong
  periodicity.  For the latter case, the equality should be
  read as $\beta$-convertibility.  For the former, it should be
  read as \boehm{} Tree equality.

  ($\RA$)  Suppose that $M = C[M_1]\cdots[M_n]$,
  with $C\cxthole_1\cdots\cxthole_n$ a B\"ohm Tree context.
  Let $x$ be an infinite path through $\bt{M}$.  Since
  $C\cxthole_1\cdots\cxthole_n$ is finite,
  $x$ must pass through one of the holes $\cxthole_i$.
  Since $M=M_i$ is periodic at this position,
  $x$ passes through a periodic position.
  Since $x$ was arbitrary, $M$ is locally periodic.

  \newcommand{\bsp}{\mathcal X}
  ($\LA$)  Suppose that $M$ is locally periodic,
  and let $\bsp$ denote the set of infinite paths
  through the B\"ohm Tree of $M$
  ($\bsp$ could be called the \emph{B\"ohm space} of $M$).
  Since $\bt{M}$ is a finitely branching tree,
  the space $\bsp$ is compact, when given the topology
  generated by the cylinders
  \[X_\sigma = \setof{\theta \where \sigma \prefixeq \theta}\]

  For $x \in \bsp$, let $\sigma_x$ be a periodic position
  that $x$ passes through (by local periodicity).
  Then the collection $\Sigma = \setof{\sigma_x}_{x \in \bsp}$
  is an open cover of $\bsp$.  By compactness, there exists
  a subcollection $\sigma_1, \dots, \sigma_n \in \Sigma$
  which covers $\bsp$.  Let $C \equiv C\cxthole_1\cdots\cxthole_n$ be
  obtained from $M$ by unfolding its B\"ohm Tree up to depth
  $\max\setof{\len{\sigma_1},\dots,\len{\sigma_n}}$ and placing holes
  $\cxthole_i$ at positions $\sigma_i$.  We claim that $C$ is
  the desired B\"ohm Tree context.
  
  By construction, $C$ is a B\"ohm Tree context
  with $M = C[M_{\sigma_1}]\cdots[M_{\sigma_n}]$.
  We also have that $M$ is periodic at $\sigma_i$,
  hence $M = M_{\sigma_i}$.
  Furthermore, if $x \in \bsp$,
  then $x$ passes through $\sigma_i$ for some $i$.  Therefore,
  outside the holes $\cxthole_i$ the B\"ohm Tree of $C$ is
  finite.
  
  Indeed, $M$ is globally periodic.
\end{proof}

In light of the proposition above,
when $M$ is either locally or globally periodic,
we simply say that $M$ is periodic.

The resulting notion of periodicity is somewhat restrictive,
because it does not account for terms in which
repetition begins later in the B\"ohm Tree.
For example if $Y$ is an fpc, then $Yx$ is periodic,
but $Y$ itself is not, because the first node of its
B\"ohm Tree has the abstraction $\lambda f. f$, which
is never repeated in the infinite sequence of $f$'s that
follows.

This situation motivates the following definition.

\begin{defi}
  Suppose that $\sigma \in \pos{\bt{M}}$ and $\theta \prefixeq \sigma$.
  Let $\pi$ be such that $\posconcat{\theta}{\pi} = \sigma$.
  \begin{enumerate}
    \item
      \emph{$M$ is periodic at $\sigma$ with offset $\theta$} if
      $M_\theta$ is periodic at $\pi$.  
      In this case, we say that $\sigma$ is \emph{cyclic} in $\bt{M}$, 
      and call $\pi$ the \emph{period} and $\theta$ the \emph{phase}
      of the \emph{cycle at $\sigma$}.
    \item
      $M$ is \emph{weakly periodic at $\sigma$ with
      offset $\theta$} if $M_\theta$ is weakly periodic at
      $\pi$.  We say that $\sigma$ is \emph{weakly cyclic}, and
      call $\pi$ the period and $\theta$ the phase of the
      \emph{weak cycle at} $\sigma$.
    \item
      $M$ is \emph{(weakly) fully periodic} if every infinite path
      through $\bt{M}$ passes through a
      (weak) cycle $(\theta,\pi)$
      (that is, eventually comes into a periodic subterm of $M$).
  \end{enumerate}
\end{defi}

\noindent As before, it is easy to see that $M$ is fully periodic
if and only if there are finitely many positions
$\sigma_1, \dots, \sigma_n \in \pos{\bt{M}}$,
$\sigma_i = \posconcat{\theta_i}{\pi_i}$ such
that $M$ is periodic at $\sigma_i$ with offset $\theta_i$
and period $\pi_i$, and outside these positions the
B\"ohm Tree is finite.

Periodic and fully periodic terms are the perfect
targets for exercising the clocks method.
We will now illustrate this with an example, where,
instead of fpc's, we consider enumerators,
also called evaluators, for Combinatory Logic.
There are many possibilities for such an enumerator $\comb{E}$,
depending on how one does the coding.  However, in most schemes
$\comb{E}$ is a periodic term.

\begin{exa}
  Let $\code{\cdot}:CL \to \lterm$ be defined as follows:
  \begin{align*}
    \code{\cK} &= \lambda z. z\cK\cK\cI\\
    \code{\cS} &= \lambda z. z\cK\cS\cI\\
    \code{MN} &= \lambda z. z (\cK \cI) \code{M} \code{N}
  \end{align*}
Let $\seqof{M} = \lambda z. z M$, with $z \notin \fv{M}$ and $\omega = \mylam{x}{xx}$.
Then some possible evaluators for $\code{\cdot}$ include:
\begin{align}
  \label{eq:1}
  \ec_1 &\equiv \omega
  (\lambda w. \seqof{\lambda a b c. a b ((w w b) (w w c))})
  = \seqof{\lambda a b c. a b ((\ec_1 b) (\ec_1 c))}\\
  \label{eq:2}
  \ec_2 &\equiv \omega
  (\lambda w. \seqof{\lambda a b c. a b (\cS \seqof{b} \seqof{c} (w w))})
  \,= \seqof{\lambda a b c. a b (\cS \seqof{b} \seqof{c} \ec_2)}\\
  \label{eq:3}
  \ec_3 &\equiv
  \seqof{\omega (\lambda w a b c. a b (\cS b c (w w)))}
  = \seqof{\ec'}, \quad\textnormal{where }
  \ec' = \lambda a b c. a b (\cS b c \ec')
\end{align}

It is straightforward to verify that,
for $i \in \setof{1,2,3}$ we have $\ec_i \code{M} = M$ for
any $\setof{\cK,\cS}$-term $M$.  The evaluators look similar,
and their B\"ohm Trees are indeed the same.
However, it is not immediate
whether the terms are $\beta$-convertible or if they are distinct.

Fortunately, since these terms are periodic, their clocked \boehm{} Trees
can be computed easily:
\begin{enumerate}
  \item
    For $N \in \lterm$, let
    $V_N \equiv \lambda a b c. a b ((N N b) (N N c))$, and
    $W = \lambda w. \seqof{V_w}$.
    Then we have
    \begin{align*}
      \ec_1 \to_h W W
      &\to_h \lambda z. z V_W \\
      &\equiv \lambda z. z (\lambda a b c. a b ((W W b) (W W c)))\\
      &\to_h \lambda z. z (\lambda a b c. a b
        ((\lambda z'. z' V_W) b (W W  c)))\\
      &\to_h \lambda z. z (\lambda a b c. a b
         (b V_W (W W c)))\\
      &\to \lambda z. z (\lambda a b c. a b
         (b V_W ((\lambda z'. z' V_W) c)))\\
      &\to \lambda z. z (\lambda a b c. a b
         (b V_W (c V_W)))\\
      &\to \cdots
    \end{align*}
    Hence $\ec_1$ is periodic at positions $02000212$ and $02000222$ 
    (in hnf-notation $010$ and $0110$ respectively)
    with offset $02$ and periods $000212$ and $000222$, 
    which allows us to construct the clocked \boehm{} Tree of $\ec_1$.
  
  \item 
    For $N \in \lterm$, let
    $V_N \equiv \lambda a b c. a b (\cS \seqof{b} \seqof{c} (ww))$,
    $W = \lambda w. \seqof{V_w}$,
    $\cS'_N = \lambda y z. N z (y z)$,
    and $\cS''_{N,M} = \lambda z. N z (M z)$.
    Then we have
    \begin{align*}
      \ec_2 \to_h WW
      &\to_h \lambda z. z V_W\\
      &\equiv \lambda z. z (\lambda a b c. a b
                   (\cS \seqof{b} \seqof{c} (WW)))\\
      &\to_h \lambda z. z (\lambda a b c. a b
                   (\cS'_{\seqof{b}} \seqof{c} (WW)))\\
      &\to_h \lambda z. z (\lambda a b c. a b
                   (\cS''_{\seqof{b},\seqof{c}} (WW)))\\
      &\to_h \lambda z. z (\lambda a b c. a b
                   (\seqof{b} (WW) (\seqof{c} (WW))))\\
      &\to_h \lambda z. z (\lambda a b c. a b
                   (W W b (W W c)))\\
      &\to^2 \lambda z. z (\lambda a b c. a b
                   ((\lambda z'. z' V_W) b ((\lambda z'. z' V_W) c)))\\
      &\to^2 \lambda z. z (\lambda a b c. a b
                   ((b V_W) (c V_W)))\\
      &\to \cdots
    \end{align*}
    Again, we are at a point where every position lies on a cycle, giving
    us a full description of the clocked \boehm{} Tree of $\ec_2$.

  \item
    Let $V_N = \lambda a b c. a b (\cS b c (N N))$,
    and $W$, $\cS'_N$, $\cS''_{N,M}$ as before.  Then
    \begin{align*}
      \ec_3 \equiv \seqof{\omega W} \to_h \seqof{W W}
      &\to_h \lambda z. z V_W\\
      &\equiv \lambda z. z (\lambda a b c. a b (\cS b c (W W)))\\
      &\to_h \lambda z. z (\lambda a b c. a b (\cS'_b c (W W)))\\
      &\to_h \lambda z. z (\lambda a b c. a b (\cS''_{b,c} (W W)))\\
      &\to_h \lambda z. z (\lambda a b c. a b (b (W W) (c (W W))))\\
    \end{align*}
\end{enumerate}

\begin{figure}[ht!]
\begin{center}
  \begin{tikzpicture}[level distance=10mm,inner sep=0.5mm,
                      level 1/.style={sibling distance=36mm},
                      level 2/.style={sibling distance=18mm},
                      level 3/.style={sibling distance=12mm},
                      level 4/.style={sibling distance=4.5mm}
                     ]
    \node  {$\mylam{z}{z}$} \annotatednode{$\mylam{z}{z}$}{2}
      child { node (vw) {$\mylam{abc}{a}$} \annotatednodeb{$\mylam{abc}{a}$}{0}{xshift=3mm}
        child { node {$b$} \annotatednode{$b$}{0} }
        child { node (b) {$b$} \annotatednode{$b$}{2}
          child[grow=-45,level distance=12mm] { node (c) {$c$} \annotatednode{$c$}{2} }
        }
      };
    \draw (b) to[out=-130,in=150,looseness=4.5] (vw);
    \draw (c) to[out=-70,in=20,looseness=3] (vw);
  \end{tikzpicture}
  \begin{tikzpicture}[level distance=10mm,inner sep=0.5mm,
                      level 1/.style={sibling distance=36mm},
                      level 2/.style={sibling distance=18mm},
                      level 3/.style={sibling distance=12mm},
                      level 4/.style={sibling distance=4.5mm}
                     ]
    \node  {$\mylam{z}{z}$} \annotatednode{$\mylam{z}{z}$}{2}
      child { node (vw) {$\mylam{abc}{a}$} \annotatednodeb{$\mylam{abc}{a}$}{0}{xshift=3mm}
        child { node {$b$} \annotatednode{$b$}{0} }
        child { node (b) {$b$} \annotatednode{$b$}{6}
          child[grow=-45,level distance=12mm] { node (c) {$c$} \annotatednode{$c$}{2} }
        }
      };
    \draw (b) to[out=-130,in=150,looseness=4.5] (vw);
    \draw (c) to[out=-70,in=20,looseness=3] (vw);
  \end{tikzpicture}
  \vspace{-2cm}
  \caption{Clocked BT's for $\ec_1$ and $\ec_2$, in hnf-notation.}
  \label{fig:e1e2}
\end{center}
\end{figure}
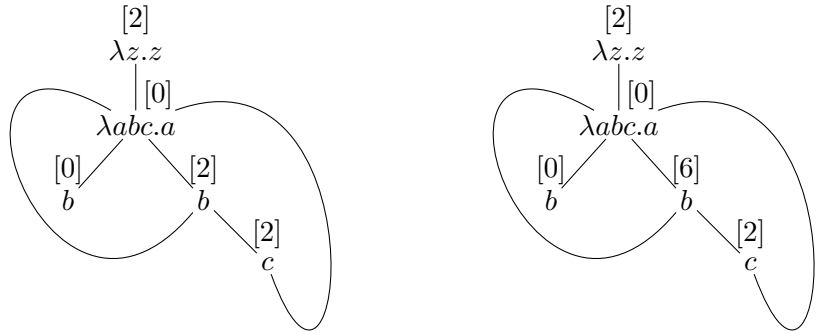

\begin{figure}[ht!]
\begin{center}
  \begin{tikzpicture}[level distance=10mm,inner sep=0.5mm,
                      level 1/.style={sibling distance=36mm},
                      level 2/.style={sibling distance=18mm},
                      level 3/.style={sibling distance=18mm},
                      level 4/.style={sibling distance=18mm}
                     ]
    \node  {$\mylam{z}{z}$} \annotatednode{$\mylam{z}{z}$}{0}
      child { node {$\mylam{abc}{a}$} \annotatednodeb{$\mylam{abc}{a}$}{2}{xshift=3mm}
        child { node {$b$} \annotatednode{$b$}{0} }
        child { node (b) {$b$} \annotatednode{$b$}{3}
          child[level distance=15mm] { node (vw_left) {$\mylam{abc}{a}$} \annotatednodeb{$\mylam{abc}{a}$}{1}{xshift=-2mm}
            child { node {$b$} \annotatednode{$b$}{0} }
            child { node (b_left) {$b$} \annotatednode{$b$}{3}
              child[grow=-45,level distance=12mm] { node (c_left) {$c$} \annotatednode{$c$}{0} }
            }
          }
          child[grow=-45,level distance=12mm] { node (c) {$c$} \annotatednode{$c$}{0} }
        }
      };
    \draw (b_left) to[out=-130,in=-180,looseness=4] (vw_left);
    \draw (c_left) to[out=-70,in=0,looseness=3] (vw_left);
    \draw (c) to[out=-150,in=45,looseness=2] (vw_left);
  \end{tikzpicture}
  \vspace{-1.9cm}
  \caption{Clocked BT for $\ec_3$, in hnf-notation.}
  \label{fig:e3}
\end{center}
\end{figure}
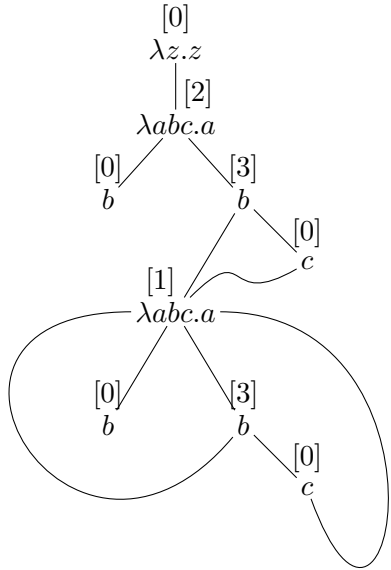

\noindent The \boehm{} Trees of $\ec_i$ are displayed in Figures~\ref{fig:e1e2} and \ref{fig:e3}.
This information, together with
Theorem~\ref{thm:simple}, allows us to discriminate between the terms.
Note that $\ec_1$ does not eventually improve on $\ec_3$: 
the occurrences of the variable $c$ in the clocked \boehm{}
Tree of $\ec_3$ are annotated by a $0$, while the corresponding
subterms of $\ec_1$ take two head steps to converge.  Furthermore,
the head redexes contracted in this reduction
$WW$ and $\seqof{V_W}c$ are both call-by-value.
By Theorem~\ref{thm:simple} we conclude that $\ec_3 \neq_\beta \ec_1$.
In contrast, $\ec_2$
\emph{is} improved by $\ec_1$.  In fact, $\ec_1$ and
$\ec_2$ can be reduced so that their clocked \boehm{}
Trees match eventually.  Upon further inspection, we see that
the two terms are indeed convertible.
\end{exa}

\begin{rem}
We note that (fully) periodic terms are
the $\lambda$-calculus analogue of the usual `circular'
terms defined by mutually recursive definitions.
The {\tt let} binder in Haskell is an example of
such a construct --- one difference being that, because
the semantics of Haskell is graph rewriting, such
terms actually denote possibly cyclic graphs rather
than their unfoldings as infinite trees.
For example, the following is an example of a
`periodic Haskell term':

\begin{samepage}
\begin{small}
\begin{verbatim}
  streamSys =
    let x = f 0 y
            where f n (b:ys) = n : f (b n) ys
        y = f id z
            where f g (c:cs) = g : f (\n -> g c + n) cs
        z = zz 0
            where zz n = n : zz (n+1)
    in x
  
  *Main> take 20 streamSys
  [0,0,0,1,4,10,20,35,56,84,120,165,220,286,364,455,560,680,816,969]
\end{verbatim}
\end{small}
\end{samepage}

\noindent In discussion of such recursive systems, a well-known topic is the
distinction between `strong equality' and `weak equality';
the former being syntactic convertibility between two such
systems, and the latter being bisimulation or
`behavioral equivalence'.  
For this reason, the complexity of strong
equality is usually $\Sigma_1$, since one only needs to
check the existence of a proof object (the conversion sequence)
for the equality between two terms, while weak equality is
$\Pi_2$, since it asks whether for every number $n$, there
is a conversion that makes the terms coincide up to depth $n$.

This discussion applies to our notions of fpc's and wfpc's,
and also periodic and weakly
periodic terms.  From this point of view,
the clocked B\"ohm Trees offer a refinement of these two
types of equality, replacing dichotomy with a spectrum,
or hierarchy, of equality,
by augmenting the syntactic shape of the terms
with the information of how quickly the
computation they define is performed.
For the case of $\lambda$-calculus, this spectrum was
displayed in Figure~\ref{fig:overivew}.
\end{rem}

Concluding this elaboration on cyclic terms, we note that
also for the fragment of $\lambda$\nb-calculus consisting of
$\mu$-terms, with the definition of the corresponding
$\mu$-rule $\mu x. A(x) \to A(\mu x. A(x))$,
the notion of clocked \boehm{} Trees may be interesting,
as it yields an equality strictly in between weak and strong
equality, as it is called in~\cite{endr:grab:klop:oost:2011,bare:dekk:stat:2011}.  
There $\mu$-terms are treated extensively,
for their application as a representation of recursive types.

\section{Atomic Clocks}\label{sec:atom}

We have introduced clocked \bohm{} Trees for discriminating $\lambda$-terms.
In this section, we refine the clocks to measure not only the \emph{number} of head
steps, but, in addition, the \emph{position} of each of these steps.
We call these clocks `atomic'.
We write $\hredat{p}$ for the head reduction step at position $p$.

Before we give the formal definition (Definition~\ref{def:abohm}), we consider a motivating example.
We discriminate $\fpc{2}$ from $\comb{U}_2$.
First, we reduce both terms to simple reducts:
\begin{align*}
  \fpc{2} x &\equiv \fpcC \delta \delta x \mred \eta \eta \delta x 
  && \text{where $\eta = \mylam{ab}{b(aab)}$}\\
  \comb{U}_2 x &\equiv \fpcC (\cS \cS) \cI x \mred \theta \theta \cI x 
  && \text{where $\theta = \mylam{abc}{bc(aabc)}$}
\end{align*}
Second, we compute the atomic clocks of these simple reducts, 
that is, the positions of the head steps, as follows:
\begin{align*}
  \eta \eta \delta x
  &\hredat{11} (\mylam{b}{b(\eta\eta b)}) \delta x
  \hredat{1} \delta (\eta\eta \delta) x\\
  &\hredat{1} (\mylam{b}{b (\eta\eta \delta b)}) x
  \hredat{\posemp} x (\eta\eta \delta x)
  \\
  \theta \theta \cI x
  &\hredat{11} (\mylam{bc}{bc(\theta\theta bc)}) \cI x
  \hredat{1} (\mylam{c}{\cI c(\theta\theta \cI c)}) x\\
  &\hredat{\posemp} \cI x(\theta\theta \cI x)
  \hredat{1} x(\theta\theta \cI x)
\end{align*}
Thus the atomic clocks of these terms are:
\begin{align*}
  \abohm{\eta \eta \delta x} &= \annotate{11,1,1,\posemp} (x \abohm{\eta \eta \delta x})\\
  \abohm{\theta \theta \cI x} &= \annotate{11,1,\posemp,1} (x \abohm{\theta \theta \cI x})
\end{align*}
Note that both terms have the (non-atomic) clocked \bohm{} Tree $T \equiv \annotate{4} (x T)$.
Hence the method from the previous section is not applicable.
However, the atomic clocks do allow us to discriminate the terms. Hence $\fpc{2} \nconv \comb{U}_2$
(by Corollary~\ref{cor:simple:simple} which generalises to the setting of atomic BT's).
Note that the (non-atomic) clocked BT's can be obtained by
taking the length of the lists of positions.

For lists $\vec{p}, \vec{q}$ of positions, we write $\vec{p} \cdot \vec{q}$
for concatenating $\vec{p}$ to $\vec{q}$.
We write $\hredat{\tuple{p_1,\ldots,p_n}}$ for the rewrite sequence
$\hredat{p_1} \cdots \hredat{p_n}$ consisting of steps at position $p_1$,\ldots,$p_n$.

\begin{defi}[Atomic clock \bohm{} Trees]\label{def:abohm}
  Let $M \in \lterm$.
  The \emph{atomic clock \bohm{} Tree $\abohm{M}$ of $M$}
  is an annotated infinite term defined as follows.
  If $M$ has no hnf, then define $\abohm{M}$ as $\sink$.
  Otherwise,
  there is a head reduction 
  \begin{align*}
    M \hredat{p_1} \cdots \hredat{p_k} \mylam{x_1}{\ldots\mylam{x_n}{y M_1 \ldots M_m}}
  \end{align*}
  of length $k$ to hnf.
  Then we define 
  $\abohm{M}$ as the term:
  \begin{align*}
    \abohm{M} 
    = \annotate{\tuple{p_1,\ldots,p_k}}{\mylam{x_1}{\ldots\mylam{x_n}{y \abohm{M_1} \ldots \abohm{M_m}}}}
  \end{align*}
\end{defi}

The theory developed for (non-atomic) BT's in Section~\ref{sec:clocked}
generalises to atomic trees, as follows.

\begin{thm}\label{thm:lift}
  For lists of positions $\vec{p}, \vec{q}$ we define
  $\vec{p} \ge \vec{q}$ whenever $\vec{q}$ is a subsequence of $\vec{p}$,
  and
  $\vec{p} > \vec{q}$ if additionally $\vec{p} \ne \vec{q}$.
  Here $\tuple{a_1,\ldots,a_n}$ is a subsequence of $\tuple{b_1,\ldots,b_m}$
  if there exist indexes $i_1 < i_2 < \ldots < i_n$ such that
  $\tuple{a_1,\ldots,a_n} = \tuple{b_{i_1},\ldots,b_{i_n}}$.
  
  Using this notation for comparing the atomic annotations
  (lists of positions),
  Proposition~\ref{prop:clocks}, Theorem~\ref{thm:general},
  Proposition~\ref{prop:simple}, Theorem~\ref{thm:simple}, and Corollary~\ref{cor:simple:simple}
  remain valid.
\end{thm}

As an application of using atomic clocks to discriminate $\lambda$-terms,
we show that every combination of 
the fpc-generating vectors $\leftappiterate{\cxthole(\cS\cS)}{\cS}{n}\cI$ 
from Theorem~\ref{thm:scott:sequence} applied to Curry's fpc $\fpcC$
gives rise to inconvertible fpc's.
We note that this cannot be proved using non-atomic clocks,
as for example we have 
$\cbohm{\leftappiterate{\leftappiterate{\fpcC(\cS\cS)}{\cS}{n}\cI(\cS\cS)}{\cS}{m}\cI} 
= \cbohm{\leftappiterate{\leftappiterate{\fpcC(\cS\cS)}{\cS}{m}\cI(\cS\cS)}{\cS}{n}\cI}$.

\begin{prop}\label{prop:scott:free}
  Let $\comb{G}_n = \leftappiterate{\cxthole(\cS\cS)}{\cS}{n}\cI$ 
  the fpc-generating vectors from Theorem~\ref{thm:scott:sequence}.
  For $n_1,\ldots,n_k \in \nat$ we define
  \begin{align*}
    \comb{Y}^{\tuple{n_1,\ldots,n_k}} = \comb{G}_{n_k}[\ldots \comb{G}_{n_1}[\fpcC]\ldots]
  \end{align*}
  All these fpc's are inconvertible, that is,
  $\vec{n} \ne \vec{m}$ implies $\comb{Y}^{\vec{n}} \nconv \comb{Y}^{\vec{m}}$.
\end{prop}

\begin{proof}
  In this proof we abbreviate context filling by simply concatenating the contexts,
  so that for $\comb{Y}^{\tuple{n_1,\ldots,n_k}}$ defined above we have $\comb{Y}^{\tuple{n_1,\ldots,n_k}} = \fpcC \comb{G}_{n_1} \ldots \comb{G}_{n_k}$.

  For $\vec{p} = \tuple{1^{m_1},\ldots,1^{m_k}}$ a list of positions
  we define $\posexp{\vec{p}}{n}$ as follows:
  \begin{align*}
    \posexp{\vec{p}}{0} &= \tuple{} \\
    \posexp{\vec{p}}{n} &= \vec{p} \cdot (\posexp{\tuple{1^{m_1-1},\ldots,1^{m_k-1}}}{(n-1)}) && (n > 0) 
  \end{align*}
  For example, we have $\tuple{1^5} \times 3 = \tuple{1^5,1^4,1^3}$.

  We also use the following abbreviations, for $n\in\nat$:
  \begin{align*}
    \comb{G}'_n & = \leftappiterate{\cxthole}{\cS}{n}\cI 
    &
    \theta & = \mylam{abc}{bc(aabc)} 
    \\
    \nf{\comb{G}_n} &= \leftappiterate{\cxthole \nfSS}{\cS}{n}\cI 
    &
    \nfSS & = \nf{\cS\cS} = \mylam{abc}{bc(abc)} 
  \end{align*}
  (Note that $\omega_{\nfSS} \to_\beta \theta$.)
  For $n_1,\ldots,n_k \in \nat$ we define an fpc $\comb{Y}$ by
  \begin{align*}
    \comb{Y} & = \theta \theta \comb{G}'_{n_1} \nf{\comb{G}_{n_2}} \ldots \nf{\comb{G}_{n_k}}
  \end{align*}
  so that we have the following reduction: 
  \begin{align*}
    \comb{Y}^{\tuple{n_1,\ldots,n_k}} x 
    \equiv \fpcC \comb{G}_{n_1} \ldots \comb{G}_{n_k}
    \mred \comb{Y} x 
    && (k \ge 1)
  \end{align*}
  and we observe that $\comb{Y}x$ is a simple term 
  (as can be inferred from the reductions below).

  Let $V_m = \cxthole N_1\,\ldots\,N_m$ be a vector of length $m$.
  Then we have, for every $n \ge 0$:
  \begin{align*}
    \theta\theta\comb{G}'_{n}V_m
    & \equiv \leftappiterate{\theta\theta}{\cS}{n} \cI V_m \\
    & \hredat{\posexp{(\posexp{\tuple{1^{n+m+1}}}{3})}{n}} 
      \cS\cI(\leftappiterate{\theta\theta}{\cS}{n}{\cI}) V_m \tag*{$(\dagger)_{n,m}$} \\
    & \hredat{\posexp{\tuple{1^{m+1}}}{2}} 
      (\mylam{c}{\cI c(\leftappiterate{\theta\theta}{\cS}{n}{\cI} c)}) V_m
  \end{align*}
  For $M$ an arbitrary term, we have the following reductions $(*)_{n,m}$ ($n,m\ge 0$):
  \begin{align*}
    (\mylam{c}{\cI c(Mc)}) \nf{\comb{G}_n} V_m 
    & \equiv
    \leftappiterate{ (\mylam{c}{\cI c(Mc)}) \nfSS}{\cS}{n} \cI V_m
    \\
    &\hredat{\tuple{1^{n+m+1}}}
    \cI \nfSS \leftappiterate{ (M \nfSS)}{\cS}{n} \cI V_m
    \\
    &\hredat{\tuple{1^{n+m+2}}}
    \nfSS \leftappiterate{ (M \nfSS)}{\cS}{n} \cI V_m
    \tag*{$(*)_{n,m}$} \\
    &\hredat{\posexp{(\posexp{\tuple{1^{n+m+1}}}{3})}{n}}
    \cS \cI (M \nf{\comb{G}_n}) V_m
    \\
    &\hredat{\posexp{\tuple{1^{m+1}}}{2}}
    (\mylam{c}{\cI c (M \comb{G}_n c)}) V_m
  \end{align*}
  Moreover, let $(\ddagger)$ denote the following rewrite sequence:
  \begin{align*}
    (\mylam{c}{\cI c (\comb{Y} c)}) x 
    \hredat{\tuple{1^0,1^1}}
    x (\comb{Y} x)
    \tag*{$(\ddagger)$}
  \end{align*}
  The rewrite sequence $\comb{Y}x \mhred x(\comb{Y}x)$ 
  is composed of $k$ subsequences 
  as follows: 
  \begin{align*}
    (\dagger)_{n_1,m_1},(*)_{n_2,m_{2}}, (*)_{n_3,m_3}, \ldots, (*)_{n_k,m_k},(\ddagger)
    \tag*{$(*)$}
  \end{align*}
  where $m_i$ is defined by
  $m_k = 1$ and $m_i = m_{i+1} + n_{i+1} + 2$ ($1 \le i < k$).
  For a rewrite sequence $\sigma$, let $\tuple{\sigma}$ denote
  the sequence of positions of the steps in $\sigma$.
  Note that the atomic clock \bohm{} Tree of $\comb{Y}^{\vec{n}}x$
  is of the form 
  $\abohm{\comb{Y}^{\vec{n}}x} = \annotate{\tuple{(*)}}{\;x(\annotate{\tuple{(*)}}{\;x(\ldots)})}$.

  As the goal is to prove inconvertibility, we may without loss of generality
  assume that $n_k \ne 0$ since the fact $M N \ne_\beta M' N \implies M \neq_\beta M'$
  allows us to append $\comb{G}_{n}$ with arbitrary $n > 0$.
  We argue that the starting positions of the subsequences $\tuple{(*)_{n_i,m_i}}$ in $\tuple{(*)}$
  coincide with the occurrences of subsequences of the form (\textmusicalnote)
  $1^l$, $1^{l+1}$, $1^{l}$, $1^{l-1}$, $1^{l-2}$ (for some $l \in \nat$), that is,
  an increment followed by four decrements.
  (In particular, we find exactly the patterns for $l = n_i + m_i+1$.)
  This suffices to derive $k,n_1,\ldots,n_k$ 
  since then the number of occurrences (\textmusicalnote) is $k-1$,
  and the $n_i$'s are a function of the length of the blocks; see also Example~\ref{ex:scott:free} below. 
  This shows that $\vec{n} \ne \vec{m}$ implies that 
  $\neg(\abohm{\comb{Y}^{\vec{n}}} \relev{=} \abohm{\comb{Y}^{\vec{m}}})$,
  and hence we conclude $\comb{Y}^{\vec{n}} \nconv \comb{Y}^{\vec{m}}$ 
  by atomic version of Corollary~\ref{cor:simple:simple}, see Theorem~\ref{thm:lift}.
    
  For the sequences $\tuple{(*)_{n_i,m_i}}$ where $n_i > 0$,
  notice that $\tuple{(*)_{n_i,m_i}}$ starts with the positions
  $1^{n_i+m_i+1},1^{n_i+m_i+2}$, $1^{n_i+m_i+1}$, $1^{n_i+m_i}$, $1^{n_i+m_i-1}$, 
  a subsequence of the form (\textmusicalnote);
  this is the only occurrence of four consecutive decrements in $\tuple{(*)_{n_i,m_i}}$.
  For sequences $\tuple{(*)_{n_i,m_i}}$ with $n_i = 0$,
  we have $i < k$ (since $n_k \ne 0$),
  and $\tuple{(*)_{n_i,m_i}}$ is of the form $1^{m_i+1},1^{m_i+2},1^{m_i+1},1^{m_i}$.
  This combined with the first element $m_i - 1 = m_{i+1} + n_{i+1} + 1$ of $\tuple{(*)_{n_{i+1},m_{i+1}}}$
  is an occurrence of the form (\textmusicalnote).
  Finally, we need to check that there are no other occurrences of (\textmusicalnote) in $\tuple{(*)}$.
  Note that other occurrences of four consecutive decrements can only
  occur as overlaps between  $\tuple{(*)_{n_i,m_i}}$ and  $\tuple{(*)_{n_{i+1},m_{i+1}}}$.
  Each of the sequences $\tuple{(*)_{n_{i+1},m_{i+1}}}$ starts with an increment 
  $1^{n_{i+1}+m_{i+1}+1}$, $1^{n_{i+1}+m_{i+1}+2}$, thus only the first element can 
  overlap.
  For $n_i = 0$, we have already analyzed the overlap, and for $n_i > 1$,
  only the last three elements of $\tuple{(*)_{n_i,m_i}}$ are not decreasing.
  This concludes the proof.
\end{proof}

\begin{exa}\label{ex:scott:free}
  Let $\comb{Y} = \theta \theta \comb{G}'_{2} \nf{\comb{G}_{0}} \nf{\comb{G}_{1}}$ 
  so that $\comb{Y}^{\tuple{2,0,1}} x \mred \comb{Y} x$ as in the above proof. 
  According to the calculations in the proof, 
  the atomic clock of the simple term $\comb{Y} x$ is 
  \newlength{\lul}
  \settowidth{\lul}{$1^3$}
  \[ \tuple{1^9,1^8,1^7,1^8,1^7,1^6,1^7,1^6,\underbrace{\mit{1^5,1^6,1^5,1^4},\hphantom{1^3}}_{\text{(\textmusicalnote)}}\hspace{-\lul}\overbrace{1^3,1^4,1^3,1^2,1^1}^{\text{(\textmusicalnote)}},1^2,1^1,\mit{1^0,1^1}} \]
  where we have indicated the occurrences of (\textmusicalnote),
  and have alternatingly used italics to separate the blocks $(\dagger)_{2,6}$, $\mit{(*)_{0,4}}$, $(*)_{1,1}$, and $\mit{(\ddagger)}$.
\end{exa}

\section{Clocked \levy{} and \ber{} Trees}\label{sec:levy}

In fact, there are three main semantics for the $\lambda$-calculus:
$\sbohm$, $\slevi$, and $\sber$; 
see \cite{abra:ong:1993,bera:intr:1996,beth:klop:vrij:2000,kenn:vrie:2003,bare:klop:2009}.
In the \bohm{} Tree semantics, a term is meaningful only if it has a hnf.
The \levy{} semantics weakens this condition to whnf's,
and thereby allows more terms to be distinguished.
The \ber{} Tree semantics is a further weakening
where only root-active terms are discarded as meaningless.

The notions from the Sections~\ref{sec:clocked} and~\ref{sec:atom} 
generalize directly to $\slevi$ and $\sber$ semantics.
We only treat the non-atomic versions here.

\begin{defi}[Clocked \levy{} Trees]\label{def:clevi}
  Let $M$ be a $\lambda$-term.
  The \emph{clocked \levy{} Tree $\clevi{M}$ of $M$}
  is an annotated potentially infinite term defined as follows.
  If $M$ has no whnf, then define $\clevi{M}$ as $\sink$.
  Otherwise,
  there exists a head rewrite sequence 
  $M \hredn{k} \mylam{x}{N}$ or $M \hredn{k} x M_1 \ldots M_m$ to whnf.
  In this case, we define
  $\clevi{M}$ as 
  $\annotate{k}{\mylam{x}{\clevi{N}}}$ or
  $\annotate{k}{x \clevi{M_1} \ldots \clevi{M_m}}$, respectively.
\end{defi}

\begin{defi}[Clocked \ber{} Trees]\label{def:cber}
  Let $M$ be a $\lambda$-term.
  The \emph{clocked \ber{} Tree $\cber{M}$ of $M$}
  is an annotated potentially infinite term defined as follows.
  If $M$ is root-active, let $\cber{M} \equiv \sink$.
  If $M \hredn{k} N$ rewrites to a root-stable term $N \equiv x$, $N \equiv \mylam{x}{P}$ or $N \equiv P Q$,
  then define $\cber{M}$ as
  $\annotate{k}{x}$, $\annotate{k}{\mylam{x}{\cber{P}}}$ or $\annotate{k}{\cber{P} \cber{Q}}$, respectively.
\end{defi}

\begin{exa}
  Consider the terms
  $M \defeq P P$ with $P \defeq \mylam{x}{\mylam{y}{xx}}$
  and
  $N \defeq Q Q$ with $Q \defeq \mylam{x}{\mylam{y}{\mylam{z}{xx}}}$.
  Then we find
  \begin{align*}
    \clevi{M} &\equiv \annotate{1}{\mylam{y}{\clevi{M}}}\\
    \clevi{N} &\equiv \annotate{1}{\mylam{y}{\annotate{0}{\mylam{z}{\clevi{N}}}}}
  \end{align*}
  Thus, in $\clevi{M}$ every $\lambda$ requires one head reduction step
  whereas in $\clevi{N}$ every second $\lambda$ is obtained for `free' (that is, in $0$ steps).
  
  We remark that $M$ and $N$ cannot be distinguished
  in the \bohm{} Tree semantics since $\cbohm{M} \equiv \cbohm{N} \equiv \bot$.
\end{exa}

\section{Concluding Remarks}\label{sec:conclusion}

We conclude with an encompassing conjecture, and some further research questions.
\begin{conjecture*}
  Building fpc's with fpc-generating vectors is a free construction, 
  that is, there are no non-trivial identifications.
\end{conjecture*}

\noindent
A first step is found in Intrigila's theorem $Y \delta \nconv Y$,
for any fpc $Y$.
A second step is that the \boehm{} sequence is duplicate-free.
A third step is found in our proof that the Scott sequence is duplicate-free, 
and Proposition~\ref{prop:scott:free}, which states that 
there are no identifications when starting the construction with $\fpcC$.

Other parts of the conjecture are as follows.
Let $Y,Y'$ be fpc's and $B_1 \ldots B_n, C_1 \ldots C_k$ be fpc-generating vectors.
\begin{enumerate}
  \item
  $Y \delta \conv Y' \delta$ iff $Y \conv Y'$.
  
  \item
  $Y B_1 \ldots B_n  \conv  Y' B_1 \ldots B_n$ iff $Y = Y'$.
  
  \item
    $Y B_1 \ldots B_n  \nconv Y C_1 \ldots C_k$ if
    $B_1 \ldots B_n \not\equiv C_1 \ldots C_k$.
\end{enumerate}
For general fpc's $Y$, $Y'$ these conjectures may be beyond current techniques,
but for the well-known fpc's of Curry and Turing, and the fpc-generating 
vectors introduced here, including their versions for $n > 3$, 
these problems are tractable.

Other directions of research could be
\begin{enumerate}

  \item 
    For atomic clock \boehm{} Trees (Section~\ref{sec:atom}) 
    we have recorded the positions of
    head reduction steps building up the head normal form.
    What about a generalization to other spine reduction strategies~\cite{bare:1984}?
    Would this give rise to a stronger discrimination method?

  \item 
    The notion of simple terms could be refined by focusing on an infinite
    path in the clocked \boehm{} Trees. Then duplication of redexes may be allowed
    along other paths, thereby making the method applicable to a larger 
    class of terms.
  
  \item 
    What general condition on a given term's head reduction is sufficient 
    to ensure that it possesses a `minimal clock'?  

  \item 
    Is it possible to characterize fully cyclic terms using a coinductive
    version of simple type theory?

  \item 
    Can the notion of a clock itself be given a type-theoretic interpretation?

  \item 
    Can the clocks method be used to give a quantitative measure 
    for optimization of functional programs?

  \item 
    Is it possible, using the clocks method, to extend Intrigila's result~\cite{intri:1997}, 
    and prove that there is no fixed point combinator $Y$ \emph{and} no $n\in\nat$ 
    such that $Y =_\beta \leftappiterate{Y}{\delta}{n}$?
    This is an instance of the conjecture above.

  \item
    Also interesting is to systematically study all solutions
    of equations $M \vec{x} = \vec{x} (M \vec{x})$ 
    (we might call them \emph{vector-fpc's})
    where it is understood that 
    $\vec{x} = x_1 \ldots x_n$ 
    associates to the right when it occurs at an active position 
    $\vec{x} M = x_1 (x_2 (\ldots (x_n M) \ldots))$, 
    and to the left if it occurs in a passive position 
    $M \vec{x} = M x_1 x_2 \ldots x_n$.
    Can these solutions $M$ be discriminated by the methods presented here?
    Notice that the terms~$M$ lead to fpc's; 
    for example $\leftappiterate{M}{\cI}{n-1}$ is an fpc (see Section~\ref{sec:schemes}),
    or more generally, $M \vec{N}$ with $\vec{N} =  N_1 \ldots N_{n-1}$
    is an fpc when $\vec{N} x \mred_\beta x$.

  \item 
    Another interesting notion is that of \emph{prime fpc's}, that is, 
    fpc's $Y$ not of the form $Y = Y' \vec{P}$ where $Y'$ is an fpc.
    In $\lambda$\nb-calculus no prime fpc's exist, due to $Y \conv Y (\cK Y)$,
    but in the $\lambda\cI$\nb-calculus the notion is non-trivial.

\end{enumerate}

\subsection*{Acknowledgements}
We thank Raymond Smullyan for raising our awareness
of the mystery and magic surrounding fixed point combinators,
Hans Zantema and Johannes Waldmann for communicating
some facts about Smullyan's Owl ($\comb{S}\comb{I}$), 
Henk Barendregt for intensive discussions which led us to simple terms, 
Rick Statman for his stimulating interest in this work and
suggesting some future elaboration,
and 
Gordon Plotkin for his question and subsequent communication 
pertinent to Selinger's work concerning the existence of
unorderable models of the $\lambda$-calculus.

\bibliographystyle{plain}
\bibliography{main}

\begin{thebibliography}{10}

\bibitem{abra:ong:1993}
S.~Abramsky and C.-H.L. Ong.
\newblock {Full Abstraction in the Lazy Lambda Calculus}.
\newblock {\em Information and Computation}, 105(2):159--267, 1993.

\bibitem{aehl:joac:2002}
K.~Aehlig and F.~Joachimski.
\newblock {On Continuous Normalization}.
\newblock In {\em Proc.\ Workshop on Computer Science Logic (CSL~2002)}, volume
  2471 of {\em LNCS}, pages 59--73. Springer, 2002.

\bibitem{bare:1984}
H.P. Barendregt.
\newblock {\em {The Lambda Calculus. Its Syntax and Semantics}}, volume 103 of
  {\em Studies in Logic and The Foundations of Mathematics}.
\newblock North-Holland, revised edition, 1984.

\bibitem{bare:dekk:stat:2011}
H.P. Barendregt, W.~Dekkers, and R.~Statman.
\newblock {\em {Lambda Calculus with Types}}.
\newblock Perspectives in Logic. Cambridge University Press, 2011.

\bibitem{bare:klop:2009}
H.P. Barendregt and J.W. Klop.
\newblock {Applications of Infinitary Lambda Calculus}.
\newblock {\em Information and Computation}, 207(5):559--582, 2009.

\bibitem{bera:intr:1996}
A.~Berarducci and B.~Intrigila.
\newblock {Church--Rosser $\lambda$-theories, Infinite $\lambda$-calculus and
  Consistency Problems}.
\newblock {\em Logic: From Foundations to Applications}, pages 33--58, 1996.

\bibitem{beth:2003}
I.~Bethke.
\newblock {Lambda Calculus}.
\newblock {Chapter 10 in \cite{terese:2003}}.

\bibitem{beth:klop:vrij:2000}
I.~Bethke, J.W. Klop, and R.C. de~Vrijer.
\newblock {Descendants and Origins in Term Rewriting}.
\newblock {\em Information and Computation}, 159(1--2):59--124, 2000.

\bibitem{bohm:1963}
C.~B\"{o}hm.
\newblock {The CUCH as a Formal and Description Language}.
\newblock {\em Annual Review in Automatic Programming}, 3:179--197, 1963.

\bibitem{coqu:herb:1994}
Th. Coquand and H.~Herbelin.
\newblock {$A$-Translation and Looping Combinators in Pure Type Systems}.
\newblock {\em Journal of Functional Programming}, 4(1):77--88, 1994.

\bibitem{deza:seve:vrie:2003}
M.~Dezani-Ciancaglini, P.~Severi, and F.~J. de~Vries.
\newblock {Infinitary Lambda Calculus and Discrimination of Berarducci Trees}.
\newblock {\em Theoretical Computer Science}, 2(298):275--302, 2003.

\bibitem{endr:grab:klop:oost:2011}
J.~Endrullis, C.~Grabmayer, J.W. Klop, and V.~van Oostrom.
\newblock {On Equal $\mu$-Terms}.
\newblock {\em Theoretical Computer Science}, 412(28):3175--3202, 2011.

\bibitem{endr:hend:klop:2010}
J.~Endrullis, D.~Hendriks, and J.W. Klop.
\newblock {Modular Construction of Fixed Point Combinators and Clocked B\"{o}hm
  Trees}.
\newblock In {\em Proc. Symp.\ on Logic in Computer Science (LICS~2010)}, pages
  111--119, 2010.

\bibitem{geuv:wern:1994}
H.~Geuvers and B.~Werner.
\newblock {On the Church--Rosser Property for Expressive Type Systems and its
  Consequences for their Metatheoretic Study}.
\newblock In {\em Proc. Symp.\ on Logic in Computer Science (LICS~1994)}, pages
  320--329, 1994.

\bibitem{gold:1995}
M.~Goldberg.
\newblock {Constructing Fixed-Point Combinators Using Application Survival}.
\newblock Technical Report BRICS RS-95-35, Dept.\ of Computer Science,
  University of Aarhus, 1995.

\bibitem{intri:1997}
B.~Intrigila.
\newblock {Non-Existent Statman's Double Fixed Point Combinator Does Not Exist,
  Indeed}.
\newblock {\em Information and Computation}, 137(1):35--40, 1997.

\bibitem{joac:2004}
F.~Joachimski.
\newblock {Confluence of the Coinductive $\lambda$-Calculus}.
\newblock {\em Theoretical Computer Science}, 311(1-3):105--119, 2004.

\bibitem{kenn:vrie:2003}
R.~Kennaway and F.-J. de~Vries.
\newblock {Infinitary Rewriting}.
\newblock {Chapter 12 in \cite{terese:2003}}.

\bibitem{kenn:klop:slee:vrie:1997}
R.~Kennaway, J.W. Klop, M.R. Sleep, and F.-J. de~Vries.
\newblock {Infinitary Lambda Calculus}.
\newblock {\em Theoretic Compututer Science}, 175(1):93--125, 1997.

\bibitem{klop:2007}
J.W. Klop.
\newblock {New Fixed Point Combinators from Old}.
\newblock In {\em Reflections on Type Theory, $\lambda$-Calculus, and the Mind.
  Essays Dedicated to Henk Barendregt on the Occasion of his 60th Birthday},
  pages 197--210. 2007.
\newblock Online version: \texttt{http://www.cs.ru.nl/barendregt60}.

\bibitem{plotkin:91}
G.D. Plotkin.
\newblock {A Semantics for Type Checking}.
\newblock In {\em Proc.\ Conf.\ on Theoretical Aspects of Computer Software
  (TACS 1991)}, volume 526 of {\em LNCS}, pages 1--17. Springer, 1991.

\bibitem{plot:2007}
G.D. Plotkin, 2007.
\newblock Personal communication at the symposium for {H.~Barendregt}'s 60th
  birthday.

\bibitem{sang:rutt:2012}
D.~Sangiorgi and J.J.M.M. Rutten.
\newblock {\em {Advanced Topics in Bisimulation and Coinduction}}, volume~52 of
  {\em Cambridge Tracts in Theoretical Computer Science}.
\newblock Cambridge University Press, 2012.

\bibitem{scott:1975}
D.S. Scott.
\newblock {Some Philosophical Issues Concerning Theories of Combinators}.
\newblock In C.~B\"{o}hm, editor, {\em Lambda Calculus and Computer Science
  Theory}, volume~37 of {\em LNCS}, pages 346--366, 1975.

\bibitem{seli:1996}
P.~Selinger.
\newblock {Order-Incompleteness and Finite Lambda Models}.
\newblock In {\em Proc. Symp.\ on Logic in Computer Science (LICS~1996)}, pages
  432--439, 1996.

\bibitem{seli:1997}
P.~Selinger.
\newblock {\em {Functionality, Polymorphism, and Concurrency: a Mathematical
  Investigation of Programming Paradigms}}.
\newblock PhD thesis, University of Pennsylvania, 1997.

\bibitem{seli:2003}
P.~Selinger.
\newblock {Order-Incompleteness and Finite Lambda Reduction Models}.
\newblock {\em Theorical Computer Science}, 309(1--3):43--63, 2003.

\bibitem{smull:1985}
R.~Smullyan.
\newblock {\em {To Mock a Mockingbird, and Other Logic Puzzles: Including an
  Amazing Adventure in Combinatory Logic}}.
\newblock Alfred A. Knopf, New York, 1985.

\bibitem{stat:1989}
R.~Statman.
\newblock {The Word problem for {S}mullyan's {L}ark Combinator is Decidable}.
\newblock {\em Journal of Symbolic Computation}, 7:103--112, 1989.

\bibitem{terese:2003}
Terese.
\newblock {\em {Term Rewriting Systems}}, volume~55 of {\em Cambridge Tracts in
  Theoretical Computer Science}.
\newblock Cambridge University Press, 2003.

\end{thebibliography}

\end{document}